\documentclass[journal]{IEEEtran} 

\usepackage{bbm} 

\usepackage{bm} 

\usepackage{amsmath} 

\usepackage{amsthm} 

\usepackage{amssymb} 

\usepackage{mathrsfs} 

\usepackage{optidef}  

\usepackage{algorithm} 

\usepackage[noend]{algpseudocode} 

\usepackage{cite} 

\usepackage{graphicx} 

\usepackage{textcomp} 

\usepackage{xcolor} 

\usepackage{caption} 
\usepackage{subcaption} 

\usepackage{hyperref} 

\newtheorem{theorem}{Theorem}
\newtheorem{lemma}{Lemma}

\newtheorem{corollary}{Corollary}
\newtheorem{definition}{Definition}

\newtheorem{remark}{Remark}

\newtheorem{condition}{Condition}

\newcommand\ddfrac[2]{\frac{\displaystyle #1}{\displaystyle #2}}

\DeclareMathOperator*{\argmin}{argmin}

\begin{document}

\title{Minimizing Age of Incorrect Information over a Channel with Random Delay} 

\author{
\IEEEauthorblockN{Yutao Chen and Anthony Ephremides}\\
\IEEEauthorblockA{Department of Electrical and Computer Engineering, University of Maryland}\thanks{Part of this work~\cite{b29} has been accepted to 2023 IEEE INFOCOM WKSHPS: Age of Information Workshop. }
}

\maketitle

\begin{abstract}
We consider a transmitter-receiver pair in a slotted-time system. The transmitter observes a dynamic source and sends updates to a remote receiver through an error-free communication channel that suffers a random delay. We consider two cases. In the first case, the update is guaranteed to be delivered within a certain number of time slots. In the second case, the update is immediately discarded once the transmission time exceeds a predetermined value. The receiver estimates the state of the dynamic source using the received updates. In this paper, we adopt the Age of Incorrect Information (AoII) as the performance metric and investigate the problem of optimizing the transmitter's action in each time slot to minimize AoII. We first characterize the optimization problem using the Markov decision process and investigate the performance of the threshold policy, under which the transmitter transmits updates only when the transmission is allowed and the AoII exceeds the threshold $\tau$. By delving into the characteristics of the system evolution, we precisely compute the expected AoII achieved by the threshold policy using the Markov chain. Then, we prove that the optimal policy exists and provide a computable relative value iteration algorithm to estimate the optimal policy. Furthermore, by leveraging the policy improvement theorem, we theoretically prove that, under an easily verifiable condition, the optimal policy is the threshold policy with $\tau=1$. Finally, numerical results are presented to highlight the performance of the optimal policy.
\end{abstract}


\section{Introduction}
Communication systems are used in all aspects of our lives and play an increasingly important role. As a result, communication systems are being asked to play more roles than just disseminating words, sounds, and images. With the proliferation of communication systems and the continuous expansion of their purposes, we have to demand higher performance from the communication systems. Meanwhile, we wonder whether traditional metrics such as throughput and latency can continue to meet such demands. One of the major drawbacks of such traditional metrics is that they treat every update equally and ignore that not every update can provide equally important information to the receiver for communication purposes. Because of this, researchers are trying to rethink existing communication paradigms and look for new ones, among which semantic communication is an important attempt. The semantics of information is formally defined in~\cite{b1} as the significance of the messages relative to the purpose of the data exchange. Then, semantic communication is regarded as \textit{"the provisioning of the right and significant piece of information to the right point of computation (or actuation) at the right point in time"}. Different from the classical metrics in data communication, semantic metrics incorporate the freshness of information, which is becoming increasingly important as real-time monitoring systems become ubiquitous in modern society. Typically, in such systems, a monitor monitors one or more events simultaneously and transmits status updates so that one or more remote receivers can have a good knowledge of the events. Therefore, the timeliness of information is often one of the most important performance indicators. The Age of Information (AoI), first introduced in~\cite{b1}, is one of the most successful examples of capturing information freshness. AoI tracks the time elapsed since the generation of the last received update, resulting in different treatments for different updates. For example, if the status update is significantly fresher than the information at the receiver, it will be more important and worth the extra resources to transmit. Let $V(t)$ be the generation time of the last update received up to time $t$. Then, AoI at time $t$ is defined by $\Delta_{AoI}(t) = t-V(t)$. After its introduction, AoI has attracted extensive attention~\cite{b3,b4,b5,b25}. However, AoI assumes that the age of each update always increases with time, ignoring the information content of the update. Such neglect is not always desirable. For example, in a remote monitoring system, the updates that provide the remote monitor with accurate information about the source process should be considered fresh, even if the update was generated earlier. This limitation leads to its poor performance in the problem of remote estimation. For example, we want to remotely estimate a rapidly changing event. In this case, a small AoI does not necessarily mean the receiver has accurate information about the event. Likewise, the receiver can make relatively accurate estimates without timely information if the event changes slowly.

Inspired by the above limitation, the Age of Incorrect Information (AoII) is introduced in~\cite{b6}, which combines the timeliness of updates and the information content they convey. More specifically, AoII combines the degree of information mismatch between the receiver and the source and the aging process of the mismatched information. As defined in~\cite{b6}, AoII captures the aging process of conflicting information through a time penalty function that quantifies the time elapsed since the last time the receiver had the perfect information about the source. The mismatch between the receiver's information and the source is captured by the information penalty function, which quantifies the degree of information mismatch between the two. Because of the flexibility of the penalty functions, AoII can be adapted to different systems and communication objectives by choosing different penalty functions.

Since the introduction of AoII, much work has been done to reveal its fundamental nature and performance in various communication systems. AoII minimization under resource constraints is investigated first. In~\cite{b6}, the authors investigate the minimization of AoII when there is a limit on the average number of transmissions allowed. Then, in~\cite{b7}, the authors extend the results to the generic time penalty function case. However, in both papers, the measure of the information mismatch is binary, either true or false. In~\cite{b8}, the authors investigate a similar system setting, but the AoII considers the quantified information mismatch between the source and the receiver. AoII in the context of scheduling is another critical problem. In scheduling problems, a base station observes multiple events and needs to select a subset of the users to update. Under these general settings, \cite{b9} investigates the problem of minimizing AoII when the channel state information is available and the time penalty function is generic. The authors of~\cite{b10} consider a similar system, but the base station cannot know the states of the events before the transmission decision is made. In real-life applications, we usually have no knowledge of the statistical model of the source process. Therefore, the authors in~\cite{b27} investigate the problem of minimizing AoII for an unknown Markovian source. The relationship between the estimation error and AoII is studied in~\cite{b26}. Moreover, a variant of AoII - Age of Incorrect Estimates is introduced and studied in~\cite{b28}. Although the above research covers a wide range of applications, one problem still has yet to receive much attention: minimizing AoII under delay. Communication channels usually suffer random delays in real-world applications due to various influences. Under this system setup, the authors of~\cite{b11} compare the performances of AoII, AoI, and real-time error through extensive numerical simulations. This paper considers a similar system setup, but we investigate the problem from a theoretical perspective. We accurately calculate the expected AoII achieved by some canonical policies, which enables us to theoretically solve the problem of minimizing AoII over a channel with random delay. Communication channel with a random delay has also been studied in the context of remote estimation and AoI~\cite{b13,b14,b15,b16}. However, the problem considered in this paper is very different, as AoII is a combination of age-based metric frameworks and error-based metric frameworks.

The main contributions of this paper can be summarized as follows. 1) We investigate the AoII minimization problem in a system where the communication channel suffers a random delay and characterize the optimization problem using the Markov decision process. 2) We derive the analytical expression of the expected AoII achieved by the threshold policy, under which the transmitter initiates transmission only when the transmission is allowed and AoII exceeds the threshold. 4) We prove the existence of the optimal policy and introduce a computable value iteration algorithm to estimate the optimal policy. 5) We theoretically find the optimal policy using the policy improvement theorem. 

The remainder of this paper is organized as follows. We introduce the system model and the optimization problem in Section~\ref{sec-SystemOverview}. Then, Section~\ref{sec-MDP} characterizes the problem using the Markov decision process. In Section~\ref{sec-PolicyPerformance}, we derive the analytical expression of the expected AoII achieved by the threshold policy. Then, we show the existence of the optimal policy, provide the value iteration algorithm to estimate the optimal policy, and theoretically find the optimal policy using the policy improvement theorem in Section~\ref{sec-OptimalPolicy}. Finally, Section~\ref{sec-NumericalResults} concludes the paper with numerical results that highlight the performance of the optimal policy.

\section{System Overview}\label{sec-SystemOverview}
\subsection{System Model}\label{sec-SystemModel}
We consider a slotted-time system in which a transmitter observes a dynamic source and needs to decide when to send status updates to a remote receiver so that the receiver can have a good knowledge of the current state of the dynamic source. The dynamic source is modeled by a two-state symmetric Markov chain with state transition probability $p$. The transmitter receives an update from the dynamic source at the beginning of each time slot. The update at time slot $k$ is denoted by $X_k$. The old update is discarded upon the arrival of a new one. Then, the transmitter decides whether to transmit the new update based on the current system status. When the channel is idle, the transmitter chooses between transmitting the new update and staying idle. When the channel is busy, the transmitter has no choice but to stay idle. The updates will be transmitted over an error-free communication channel that suffers a random delay. In other words, the update will not be corrupted during the transmission, but each transmission will take a random amount of time $T\in\mathbb{N}^*$. We denote the probability mass function (PMF) by $p_t\triangleq Pr(T=t)$ and assume that $T$ is independent and identically distributed for each update. When a transmission finishes, the communication channel is immediately available for the subsequent transmission.

The receiver maintains an estimate of the current state of the dynamic source and modifies its estimate each time a new update is received. We denote by $\hat{X}_k$ the receiver's estimate at time slot $k$. According to~\cite{b16}, the best estimator when $p\leq\frac{1}{2}$ is the last received update. When $p>\frac{1}{2}$, the optimal estimator depends on the realization of transmission time. In this paper, we only consider the case of $0<p\leq\frac{1}{2}$. In this case, the receiver uses the last received update as the estimate. For the case of $p>\frac{1}{2}$, the results can be extended using the corresponding best estimator. The receiver uses $ACK/NACK$ packets to inform the transmitter of its reception of the new update. As is assumed in~\cite{b6}, the transmitter receives the $ACK/NACK$ packets reliably and instantaneously because the packets are generally very small compared to the size of the status updates. When $ACK$ is received, the transmitter knows that the receiver's estimate changes to the last sent update. When $NACK$ is received, the transmitter knows that the receiver's estimate does not change. In this way, the transmitter always knows the current estimate on the receiver side.

An illustration of the system model is shown in Fig.~\ref{fig-SystemModel}.
\begin{figure}[!t]
\centering
\includegraphics[width=3in]{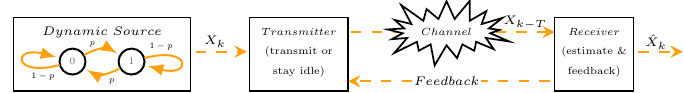}
\caption{An illustration of the system model, where $X_k$ and $\hat{X}_k$ are the state of the dynamic source and the receiver's estimate at time slot $k$, respectively.}
\label{fig-SystemModel}
\end{figure}
At the beginning of time slot $k$, the transmitter receives the update $X_k$ from the dynamic source. Then, the transmitter decides whether to transmit this update based on the system status. When the transmitter decides not to start transmission, it will stay idle. Otherwise, the transmitter will transmit the update through the communication channel, where the transmission of the update takes a random amount of time. Thus, the update received by the receiver has a delay of several time slots (i.e., $X_{k-T}$). Then, the receiver will modify its estimation $\hat{X_k}$ based on the received update and send an $ACK$ packet to inform the transmitter of its reception of the update.

\subsection{Age of Incorrect Information}\label{sec-AoII}
The system adopts the Age of Incorrect Information (AoII) as the performance metric. We first define $U_k$ as the last time slot up to time slot $k$ in which the receiver's estimate is correct. Mathematically,
\[
U_k \triangleq \max\{h:h\leq k, X_h = \hat{X}_h\}.
\]
Then, in a slotted-time system, AoII at time slot $k$ can be written as
\begin{equation}\label{eq-AoII}
\Delta_{AoII}(X_k,\hat{X}_k,k) = \sum_{h=U_k+1}^k\bigg(g(X_h,\hat{X}_h) F(h-U_k)\bigg),
\end{equation}
where $g(X_k,\hat{X}_k)$ is the information penalty function. $F(k) \triangleq f(k) - f(k-1)$ where $f(k)$ is the time penalty function. In this paper, we choose $g(X_k,\hat{X}_k) = |X_k-\hat{X}_k|$ and $f(k) = k$. Hence, $F(k)=1$ and $g(X_k,\hat{X}_k)\in\{0,1\}$ as the dynamic source has two states. Then, equation \eqref{eq-AoII} can be simplified as
\[
\Delta_{AoII}(X_k,\hat{X}_k,k) = k-U_k\triangleq\Delta_k.
\]
We can easily conclude from the simplified expression that, under the chosen penalty functions, AoII increases at the rate of $1$ per time slot when the receiver's estimate is incorrect. Otherwise, AoII is $0$. Next, we characterize the evolution of $\Delta_{k}$. To this end, we divide the evolution into the following cases.
\begin{itemize}
\item When $X_{k+1}=\hat{X}_{k+1}$, we have $U_{k+1} = k + 1$. Then, by definition, $\Delta_{k+1} = 0$.
\item When $X_{k+1}\neq\hat{X}_{k+1}$, we have $U_{k+1} = U_k$. Then, by definition, $\Delta_{k+1}=k+1-U_k=\Delta_k +1$.
\end{itemize}
Combining together, we have
\begin{equation}\label{eq-AoIIDynamics}
\Delta_{k+1} = \mathbbm{1}\{X_{k+1}\neq\hat{X}_{k+1}\}(\Delta_k+1),
\end{equation}
where $\mathbbm{1}\{A\}$ is the indicator function, whose value is one when event $A$ occurs and zero otherwise. A sample path of $\Delta_k$ is shown in Fig.~\ref{fig-SamplePath}.
\begin{figure}[!t]
\centering
\includegraphics[width=3in]{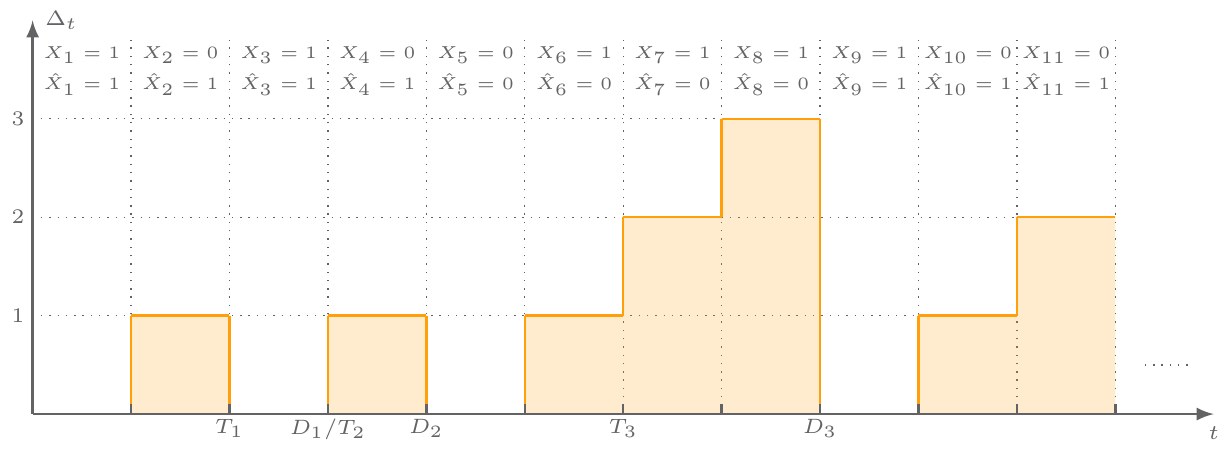}
\caption{A sample path of $\Delta_k$, where $T_i$ and $D_i$ are the transmission start and delivery time of the $i$-th update, respectively. At $T_1$, the transmitted update is $X_3$. Note that the transmission decisions in the plot are taken randomly.}
\label{fig-SamplePath}
\end{figure}
Now that the evolution of AoII has been clarified, we further discuss the system's evolution.

\subsection{System Dynamics}\label{sec-SystemDynamics}
In this subsection, we tackle the system dynamics, which will play a key role in later sections. We notice that the system's status at the beginning of time slot $k$ can be fully captured by the triplet $s_k\triangleq(\Delta_k,t_k,i_k)$ where $t_k\in\mathbbm{N}^0$ indicates the time the current transmission has been in progress. We define $t_k=0$ if there is no transmission in progress. $i_k\in\{-1,0,1\}$ indicates the state of the channel. We define $i_k=-1$ when the channel is idle, $i_k=0$ if the channel is busy and the transmitting update is the same as the receiver's current estimate, and $i_k=1$ when the transmitting update is different from the receiver's current estimate.
\begin{remark}
According to the definitions of $t_k$ and $i_k$, $i_k=-1$ if and only if $t_k=0$. In this case, the channel is idle.
\end{remark}
Then, characterizing the system dynamics is equivalent to characterizing the value of $s_{k+1}$ using $s_k$ and the transmitter's action. We use $a_k\in\{0,1\}$ to denote the transmitter's decision, where $a_k=0$ when the transmitter decides not to initiate a transmission and $a_k=1$ otherwise. Hence, the system dynamics can be fully characterized by $P_{s_k,s_{k+1}}(a_k)$, which is the probability that action $a_k$ at $s_k$ leads to $s_{k+1}$. We will revisit $P_{s_k,s_{k+1}}(a_k)$ with an in-depth analysis later.

\subsection{Problem Formulation}\label{sec-ProblemFormulation}
We define a policy $\phi$ as the one that specifies the transmitter's decision in each time slot. This paper aims to find the policy that minimizes the expected AoII of the system. Mathematically, the problem can be formulated as the following optimization problem.
\begin{argmini}|l|
{\phi \in \Phi} {\lim_{K\to\infty} \frac{1}{K}\mathbb{E}_{\phi}\left(\sum_{k=0}^{K-1}\Delta_k\right),}{\label{eq-goal}}{}
\end{argmini}
where $\mathbb{E}_{\phi}$ is the conditional expectation, given that policy $\phi$ is adopted, and $\Phi$ is the set of all admissible policies.
\begin{definition}[Optimal policy]\label{def-OptimalPolicy}
A policy is said to be optimal if it yields the minimal expected AoII.
\end{definition}
In the next section, we characterize the problem reported in \eqref{eq-goal} using a Markov Decision Process (MDP).

\section{MDP Characterization}\label{sec-MDP}
The minimization problem reported in \eqref{eq-goal} can be characterized by an infinite horizon with average cost MDP $\mathcal{M}$, which consists of the following components.
\begin{itemize}
\item The state space $\mathcal{S}$. The state $s=(\Delta,t,i)$ is the triplet defined in Section~\ref{sec-SystemDynamics} without the time stamp. For the remainder of this paper, we will use $s$ and $(\Delta,t,i)$ to represent the state interchangeably. They will synchronize any superscript or subscript.
\item The action space $\mathcal{A}$. When $i=-1$, the feasible action is $a\in\{0,1\}$ where $a=0$ if the transmitter decides not to initiate a new transmission and $a=1$ otherwise. When $i\neq-1$, the feasible action is $a=0$.
\item The state transition probability $\mathcal{P}$. The probability that the operation of action $a$ at state $s$ leads to state $s'$ is denoted by $P_{s,s'}(a)$, whose value will be discussed in the next subsection.
\item The immediate cost $\mathcal{C}$. The immediate cost for being at state $s$ is $C(s)=\Delta$.
\end{itemize}
Let $V(s)$ be the value function of state $s\in\mathcal{S}$. It is well known that the value function satisfies the Bellman equation~\cite{b19}.
\begin{equation}\label{eq-Bellman}
V(s) + \theta = \min_{a\in\mathcal{A}}\left\lbrace C(s) + \sum_{s'\in\mathcal{S}}P_{s,s'}(a)V(s')\right\rbrace \quad s\in\mathcal{S},
\end{equation}
where $\theta$ is the expected AoII achieved by the optimal policy. We will write $V(s)$ as $V(\Delta,t,i)$ in some parts of this paper to better distinguish between states. We notice that the state transition probability is essential for solving the Bellman equation. Hence, we delve into $P_{s,s'}(a)$ in the following subsection.

\subsection{State Transition Probability}\label{sec-STP}
We recall that $P_{s,s'}(a)$ is the probability that action $a$ at state $s$ will lead to state $s'$. 
Then, we define $Pr(T>k+1\mid t)$ as the probability that the current transmission will take more than $t+1$ time slots, given that the current transmission has been in progress for $t$ time slots. Hence,
\[
Pr(T>t+1\mid t) = \frac{1-Pr(T\leq t+1)}{Pr(T>t)} = \frac{1-P_{t+1}}{1-P_t},
\]
where $P_t \triangleq \sum_{k=1}^{t}p_k$. Leveraging this, $P_{s,s'}(a)$ can be obtained easily. For the sake of space, the complete state transition probabilities are detailed in Appendix~\ref{app-STP} of the supplementary material.

We notice that we do not impose any restrictions on the update transmission time, which would make the theoretical analysis very difficult and lead to long channel occupancy by a single update. Therefore, to ease the theoretical analysis and be closer to the practice, we consider the following two independent assumptions.\footnote{The results presented in this paper apply to both assumptions unless stated otherwise.}
\begin{itemize}
\item \textbf{Assumption 1}: We assume that the update will always be delivered and the transmission lasts at most $t_{max}$ time slots. More precisely, we assume $1\leq T\leq t_{max}$ and
\[
\sum_{t=1}^{t_{max}}p_t = 1,\quad p_t\geq0,\ 1\leq t\leq t_{max}.
\]
In practice, we can make the probability of the transmission time exceeding $t_{max}$ negligible by choosing a sufficiently large $t_{max}$. 
\item \textbf{Assumption 2}: We assume the transmission can last for a maximum of $t_{max}$ time slots. At the end of the $t_{max}$th time slot, the update will be discarded if not delivered, and the channel will be available for a new transmission immediately. We define $p_{t^+}\triangleq\sum_{t=t_{max}+1}^{\infty}p_t$ as the probability that the update will be discarded. In practice, similar techniques, such as time-to-live (TTL)~\cite{b22}, are used to prevent an update from occupying the channel for too long.
\end{itemize}
\begin{remark}
$t_{max}$ is a predetermined system parameter and is not a parameter to be optimized. When $t_{max}=1$, the system reduces to the one considered in~\cite{b6}, according to which the optimal policy is to transmit a new update whenever possible. Therefore, in the rest of this paper, we focus on the case of $t_{max}>1$.
\end{remark}
Under both assumptions, the transmission will last at most $t_{max}$ time slots, and the channel will be immediately available for a new transmission when the current transmission finishes. Hence, the state space $\mathcal{S}$ is reduced as $t$ is now bounded by $0\leq t\leq t_{max}-1$. Moreover, the state transition probabilities in Appendix~\ref{app-STP} of the supplementary material will be adjusted as follows.
\begin{itemize}
\item Under \textbf{Assumption 1}, updates are bound to be delivered after $t_{max}$ time slots. Hence, $Pr(T>t+1\mid t)=0$ for $t\geq t_{max}-1$.
\item Under \textbf{Assumption 2}, updates will be discarded at the end of the $t_{max}$th time slot if not delivered. Hence, $s' = (\Delta',t_{max},i')$ will be replaced by $s' = (\Delta',0,-1)$.
\end{itemize}
Having clarified the state transition probabilities, we evaluate a canonical policy in terms of the achieved expected AoII in the next section.

\section{Policy Performance Analysis}\label{sec-PolicyPerformance}
As is proved in~\cite{b6,b7,b8}, the AoII-optimal policy often has a threshold structure. Hence, we consider the threshold policy.
\begin{definition}[Threshold policy]\label{def-ThreholdPolicy}
Under threshold policy $\tau$, the transmitter will initiate a transmission only when the current AoII is no less than threshold $\tau\in\mathbbm{N}^0$ and the channel is idle.
\end{definition}
\begin{remark}
We define $\tau\triangleq\infty$ as the policy under which the transmitter never initiates any transmissions.
\end{remark}
We notice that the system dynamics under threshold policy can be characterized by a discrete-time Markov chain (DTMC). Without loss of generality, we assume the DTMC starts at state $(0,0,-1)$. Then, the state space of the Markov chain $\mathcal{S}^{MC}$ consists of all the states accessible from state $(0,0,-1)$. Since state $(0,0,-1)$ is positive recurrent and communicates with each state $s\in\mathcal{S}^{MC}$, the stationary distribution exists. Let $\pi_{s}$ be the steady-state probability of state $s$. Then, $\pi_s$ satisfies the following balance equation.
\[
\pi_s = \sum_{s'\in\mathcal{S}^{MC}}P_{s',s}(a)\pi_{s'} \quad s\in\mathcal{S}^{MC},
\]
where $P_{s',s}(a)$ is the single-step state transition probability as define in Section~\ref{sec-MDP}, and the action $a$ depends on the threshold policy. Then, the first step in calculating the expected AoII achieved by the threshold policy is to calculate the stationary distribution of the induced DTMC. However, the problem arises as the state space $\mathcal{S}^{MC}$ is infinite and intertwined. To simplify the state transitions, we recall that the transmitter can only stay idle (i.e., $a=0$) when the channel is busy. Let $\mathcal{S}^{MC}_{-1}=\{s=(\Delta,t,i):i\neq-1\}$ be the set of the state where the channel is busy. Then, for $s'\in\mathcal{S}^{MC}_{-1}$, $P_{s',s}(a) = P_{s',s}(0)$ and is independent of the threshold policy. Hence, for any threshold policy and each $s\in\mathcal{S}\setminus\mathcal{S}^{MC}_{-1}$, we can repeatedly replace $\pi_{s'}$, where $s'\in\mathcal{S}^{MC}_{-1}$, with the corresponding balance equation until we get the following equation.
\begin{equation}\label{eq-CompactBalanceEq}
\pi_{s} = \sum_{s'\in \mathcal{S}\setminus\mathcal{S}^{MC}_{-1}}P_{\Delta',\Delta}(a)\pi_{s'}\quad s\in\mathcal{S}\setminus\mathcal{S}^{MC}_{-1},
\end{equation}
where $P_{\Delta',\Delta}(a)$ is the multi-step state transition probability from state $s'=(\Delta',0,-1)$ to state $s=(\Delta,0,-1)$ under action $a$. For simplicity, we write \eqref{eq-CompactBalanceEq} as
\begin{equation}\label{eq-CompactBalanceEq2}
\pi_{\Delta} = \sum_{\Delta'\geq0}P_{\Delta',\Delta}(a)\pi_{\Delta'}\quad \Delta\geq0.
\end{equation}
As we will see in the following subsections, $\pi_\Delta$ is sufficient to calculate the expected AoII obtained by any threshold policy.
\begin{remark}
The intuition behind the simplification of the balance equations is as follows. We recall that the system dynamics when the channel is busy are independent of the adopted policy. Hence, we can calculate these dynamics in advance so that the balance equations contain only the states in which the transmitter needs to make decisions. 
\end{remark}
\noindent In the next subsection, we derive the expression of $P_{\Delta,\Delta'}(a)$.

\subsection{Multi-step State Transition Probability}\label{sec-StateTransProb}
We start with the case of $a=0$. In this case, no update will be transmitted, and $P_{\Delta,\Delta'}(0)$ is independent of the transmission delay. Then, according to Appendix~\ref{app-STP} of the supplementary material,
\[
P_{0,\Delta'}(0) = \begin{dcases}
1-p & \Delta'=0,\\
p & \Delta'=1,
\end{dcases}
\]
and for $\Delta>0$,
\[
P_{\Delta,\Delta'}(0) = \begin{dcases}
p & \Delta'=0,\\ 
1-p & \Delta' = \Delta+1.
\end{dcases}
\]
In the sequel, we focus on the case of $a=1$. We define $P^{t}_{\Delta,\Delta'}(a)$ as the probability that action $a$ at state $s=(\Delta,0,-1)$ will lead to state $s'=(\Delta',0,-1)$, given that the transmission takes $t$ time slots. Then, under \textbf{Assumption 1},
\[
P_{\Delta,\Delta'}(1) = \sum_{t=1}^{t_{max}}p_tP^t_{\Delta,\Delta'}(1).
\]
Hence, it is sufficient to obtain the expressions of $P^t_{\Delta,\Delta'}(1)$. To this end, we define $p^{(t)}$ as the probability that the dynamic source will remain in the same state after $t$ time slots. Since the Markov chain is symmetric, $p^{(t)}$ is independent of the state and can be calculated by
\[
p^{(t)} = \left(\begin{bmatrix}
1-p & p\\
p & 1-p
\end{bmatrix}^t\right)_{11},
\]
where the subscript indicates the row number and the column number of the target probability. For the consistency of notation, we define $p^{(0)}\triangleq1$. Then, we have the following lemma.
\begin{lemma}\label{lem-MSTPAss1}
Under \textbf{Assumption 1}, 
\begin{equation}\label{eq-MSTPAss1}
P_{\Delta,\Delta'}(1) = \sum_{t=1}^{t_{max}}p_tP^t_{\Delta,\Delta'}(1),
\end{equation}
where
\[
P^{t}_{0,\Delta'}(1) = 
\begin{dcases}
p^{(t)} & \Delta'=0,\\
p^{(t-k)}p(1-p)^{k-1} & 1\leq\Delta'= k\leq t,\\
0 & otherwise,
\end{dcases}
\]
and for $\Delta>0$,
\begin{multline*}
P^{t}_{\Delta,\Delta'}(1)=\\
\begin{dcases}
p^{(t)} & \Delta'=0,\\
(1-p^{(t-1)})(1-p) & \Delta'=1,\\
(1-p^{(t-k)})p^{2}(1-p)^{k-2} & 2\leq \Delta'=k\leq t-1,\\
p(1-p)^{t-1} & \Delta'=\Delta+t,\\
0 & otherwise.
\end{dcases}
\end{multline*}
Under \textbf{Assumption 1}, equation \eqref{eq-MSTPAss1} can be written equivalently as \eqref{eq-EquivalentEq1}
\begin{figure*}[!t]
\normalsize
\begin{equation}\label{eq-EquivalentEq1}
P_{\Delta,\Delta'}(1) =\begin{dcases}
\sum_{t=\Delta'}^{t_{max}}p_tP^t_{\Delta,\Delta'}(1) & 0\leq\Delta'\leq t_{max}-1,\Delta\geq\Delta',\\
\sum_{t=\Delta'}^{t_{max}}p_tP^t_{\Delta,\Delta'}(1) + p_{t'}P^{t'}_{\Delta,\Delta'}(1) & 0\leq\Delta'\leq t_{max}-1,\Delta<\Delta',\\
p_{t'}P^{t'}_{\Delta,\Delta'}(1) & \Delta'\geq t_{max}.
\end{dcases}
\end{equation}
\hrulefill
\vspace*{4pt}
\end{figure*}
where $t'\triangleq\Delta'-\Delta$ and $P^{t'}_{\Delta,\Delta'}(1)\triangleq 0$ when $t'\leq0$ or when $t'>t_{max}$. Meanwhile, $P_{\Delta,\Delta'}(1)$ possesses the following properties.
\begin{enumerate}
\item $P_{\Delta,\Delta'}(1)$ is independent of $\Delta$ when $0\leq\Delta'\leq t_{max}-1$ and $\Delta\geq\Delta'$.
\item $P_{\Delta,\Delta'}(1) = P_{\Delta+\delta,\Delta'+\delta}(1)$ when $\Delta'\geq t_{max}$ and $\Delta\geq0$ for any $\delta\geq1$.
\item $P_{\Delta,\Delta'}(1)=0$ when $\Delta'>\Delta+t_{max}$ or when $t_{max}-1<\Delta'<\Delta+1$.
\end{enumerate}
\end{lemma}
\begin{proof}
The expression of $P^t_{\Delta,\Delta'}(1)$ is obtained by analyzing the system dynamics. The complete proof can be found in Appendix~\ref{pf-CompactTrans} of the supplementary material.
\end{proof}
The state transition probabilities under \textbf{Assumption 2} can be obtained similarly. To this end, we define $P^{t^+}_{\Delta,\Delta'}(a)$ as the probability that action $a$ at state $s=(\Delta,0,-1)$ will result in state $s'=(\Delta',0,-1)$, given that the transmission is terminated. Then, we have the following lemma.
\begin{lemma}\label{lem-MSTPAss2}
Under \textbf{Assumption 2},
\begin{equation}\label{eq-MSTPAss2}
P_{\Delta,\Delta'}(1) = \sum_{t=1}^{t_{max}}p_tP^{t}_{\Delta,\Delta'}(1) + p_{t^+}P^{t^+}_{\Delta,\Delta'}(1),
\end{equation}
where
\[
P^{t}_{0,\Delta'}(1) = 
\begin{dcases}
p^{(t)} & \Delta'=0,\\
p^{(t-k)}p(1-p)^{k-1} & 1\leq\Delta'= k\leq t,\\
0 & otherwise,
\end{dcases}
\]
\[
P^{t^+}_{0,\Delta'}(1) = P^{t_{max}}_{0,\Delta'}(1),
\]
and for $\Delta>0$,
\begin{multline*}
P^{t}_{\Delta,\Delta'}(1) =\\
\begin{dcases}
p^{(t)} & \Delta'=0,\\
(1-p^{(t-1)})(1-p) & \Delta'=1,\\
(1-p^{(t-k)})p^{2}(1-p)^{k-2} & 2\leq\Delta'=k\leq t-1,\\
p(1-p)^{t-1} & \Delta'=\Delta+t,\\
0 & otherwise,
\end{dcases}
\end{multline*}
\begin{multline*}
P^{t^+}_{\Delta,\Delta'}(1) =\\
\begin{dcases}
1-p^{(t_{max})} & \Delta'=0,\\
(1-p^{(t_{max}-k)})p(1-p)^{k-1} & 1\leq\Delta'= k\leq t_{max}-1,\\
(1-p)^{t_{max}} & \Delta' = \Delta+t_{max},\\
0 & otherwise.
\end{dcases}
\end{multline*}
Under \textbf{Assumption 2}, equation \eqref{eq-MSTPAss2} can be written equivalently as \eqref{eq-EquivalentEq2}. 
\begin{figure*}[!t]
\normalsize
\begin{equation}\label{eq-EquivalentEq2}
P_{\Delta,\Delta'}(1) =
\begin{dcases}
\sum_{t=\Delta'}^{t_{max}}p_tP^t_{\Delta,\Delta'}(1)+ p_{t^+}P^{t^+}_{\Delta,\Delta'}(1) & 0\leq\Delta'\leq t_{max}-1,\Delta\geq\Delta',\\
\sum_{t=\Delta'}^{t_{max}}p_tP^t_{\Delta,\Delta'}(1) + p_{t'}P^{t'}_{\Delta,\Delta'}(1)+ p_{t^+}P^{t^+}_{\Delta,\Delta'}(1) & 0\leq\Delta'\leq t_{max}-1,\Delta<\Delta',\\
p_{t'}P^{t'}_{\Delta,\Delta'}(1)+ p_{t^+}P^{t^+}_{\Delta,\Delta'}(1) & \Delta'\geq t_{max}.
\end{dcases}
\end{equation}
\hrulefill
\vspace*{4pt}
\end{figure*}
Meanwhile, $P_{\Delta,\Delta'}(1)$ possesses the following properties.
\begin{enumerate}
\item $P_{\Delta,\Delta'}(1)$ is independent of $\Delta$ when $0\leq\Delta'\leq t_{max}-1$ and $\Delta\geq\max\{1,\Delta'\}$.
\item $P_{\Delta,\Delta'}(1) = P_{\Delta+\delta,\Delta'+\delta}(1)$ when $\Delta'\geq t_{max}$ and $\Delta>0$ for any $\delta\geq1$.
\item $P_{\Delta,\Delta'}(1)=0$ when $\Delta'>\Delta+t_{max}$ or when $t_{max}-1<\Delta'<\Delta+1$.
\end{enumerate}
\end{lemma}
\begin{proof}
The proof follows similar  steps as presented in the proofs of Lemma~\ref{lem-MSTPAss1}. The complete proof can be found in Appendix~\ref{pf-Case2TransProb} of the supplementary material.
\end{proof}
As the expressions and properties of $P_{\Delta,\Delta'}(a)$ under both assumptions are clarified, we solve for $\pi_{\Delta}$ in the next subsection.

\subsection{Stationary Distribution}\label{sec-SteadyState}
Let $ET$ be the expected transmission time of an update. Since the channel remains idle if no transmission is initiated and the expected transmission time of an update is $ET$, $\pi_{\Delta}$ satisfies the following equation.
\begin{equation}\label{eq-TotalProbs}
\sum_{\Delta=0}^{\tau-1}\pi_\Delta + ET\sum_{\Delta=\tau}^{\infty}\pi_\Delta = 1,
\end{equation}
where $ET= \sum_{t=1}^{t_{max}}tp_t$ under \textbf{Assumption 1} and $ET= \sum_{t=1}^{t_{max}}tp_t+t_{max}p_{t^+}$ under \textbf{Assumption 2}. We notice that there is still infinitely many $\pi_{\Delta}$ to calculate. To overcome the infinity, we recall that, under threshold policy, the suggested action is $a=1$ for all the state $(\Delta,0,-1)$ with $\Delta\geq\tau$. Hence, we define $\Pi\triangleq \sum_{\Delta=\omega}^{\infty}\pi_\Delta$ where $\omega \triangleq t_{max} + \tau+1$. As we will see in the following subsections, $\Pi$ and $\pi_{\Delta}$ for $0\leq\Delta<\omega-1$ are sufficient for calculating the expected AoII achieved by the threshold policy. With $\Pi$ in mind, we have the following theorem.
\begin{theorem}\label{prop-StationaryDistribution}
For $0<\tau<\infty$, $\Pi$ and $\pi_{\Delta}$ for $0\leq\Delta<\omega-1$ are the solution to the following system of linear equations.
\[
\pi_0 = (1-p)\pi_0 + p\sum_{i=1}^{\tau-1}\pi_i+ P_{1,0}(1)\left(\sum_{i=\tau}^{\omega-1}\pi_i+\Pi\right).
\]
\[
\pi_1= p\pi_0 + P_{1,1}(1)\left(\sum_{i=\tau}^{\omega-1}\pi_i+\Pi\right).
\]
\[
\Pi = \sum_{i=\tau+1}^{\omega-1}\left(\sum_{k=\tau+1}^iP_{i,t_{max}+k}(1)\right)\pi_i + \sum_{i=1}^{t_{max}}\bigg(P_{\omega,\omega+i}(1)\bigg)\Pi.
\]
\[
\sum_{i=0}^{\tau-1}\pi_i + ET\left(\sum_{i=\tau}^{\omega-1}\pi_i+\Pi\right) = 1.
\]
For each $2\leq\Delta\leq t_{max}-1$,
\begin{multline*}
\pi_\Delta =\\
\begin{dcases}
(1-p)\pi_{\Delta-1} + P_{\tau,\Delta}(1)\left(\sum_{i=\tau}^{\omega-1}\pi_i+\Pi\right) & \Delta-1<\tau,\\
\sum_{i=\tau}^{\Delta-1}P_{i,\Delta}(1)\pi_i + P_{\Delta,\Delta}(1)\left(\sum_{i=\Delta}^{\omega-1}\pi_i+\Pi\right) & \Delta-1\geq\tau.
\end{dcases}
\end{multline*}
For each $t_{max}\leq\Delta\leq\omega-1$,
\[
\pi_{\Delta} = \begin{dcases}
(1-p)\pi_{\Delta-1} & \Delta-1<\tau,\\
\sum_{i=\tau}^{\Delta-1}P_{i,\Delta}(1)\pi_i & \Delta-1\geq\tau.
\end{dcases}
\]
\end{theorem}
\begin{proof}
We delve into the definition of $\Pi$. By leveraging the structural property of the threshold policy and the properties of $P_{\Delta,\Delta'}(a)$, we obtain the above system of linear equations. The complete proof can be found in Appendix~\ref{pf-StationaryDistribution} of the supplementary material.
\end{proof}
\begin{remark}
The size of the system of linear equations detailed in Theorem~\ref{prop-StationaryDistribution} is $\omega+1$.
\end{remark}
\begin{corollary}\label{cor-StationaryDistributionSpecial}
When $\tau=0$,
\[
\pi_{0} = \frac{P_{1,0}(1)}{ET[1-P_{0,0}(1)+P_{1,0}(1)]}.
\]
For each $1\leq\Delta\leq t_{max}$,
\[
\pi_{\Delta} = \sum_{i=0}^{\Delta-1}P_{i,\Delta}(1)\pi_i + P_{\Delta,\Delta}(1)\left(\frac{1}{ET} - \sum_{i=0}^{\Delta-1}\pi_i\right).
\]
\[
\Pi = \ddfrac{\sum_{i=1}^{t_{max}}\left(\sum_{k=1}^{i}P_{i,t_{max}+k}(1)\right)\pi_i}{1-\sum_{i=1}^{t_{max}}P_{t_{max}+1,t_{max}+1+i}(1)}.
\]

When $\tau=1$,
\[
\pi_0 = \frac{P_{1,0}(1)}{pET+P_{1,0}(1)},\quad \pi_1 = \frac{pP_{1,0}(1)+pP_{1,1}(1)}{pET+P_{1,0}(1)}.
\]
For each $2\leq \Delta\leq t_{max}+1$,
\[
\pi_\Delta = \sum_{i=1}^{\Delta-1}P_{i,\Delta}(1)\pi_i + P_{\Delta,\Delta}(1)\left(\frac{1-\pi_0}{ET} - \sum_{i=1}^{\Delta-1}\pi_i\right).
\]
\[
\Pi = \ddfrac{\sum_{i=2}^{t_{max}+1}\left(\sum_{k=2}^iP_{i,t_{max}+k}(1)\right)\pi_i}{1-\sum_{i=1}^{t_{max}}P_{t_{max}+2,t_{max}+2+i}(1)}.
\]
\end{corollary}
\begin{proof}
The calculations follow similar steps as detailed in the proof of Theorem~\ref{prop-StationaryDistribution}. The complete proof can be found in Appendix~\ref{pf-AoIISpecialCase1} of the supplementary material.
\end{proof}
We will calculate the expected AoII in the next subsection based on the above results.

\subsection{Expected AoII}\label{sec-ExpectedAoII}
Let $\bar{\Delta}_{\tau}$ be the expected AoII achieved by threshold policy $\tau$. Then,
\begin{equation}\label{eq-ExpectedAoII}
\bar{\Delta}_{\tau} = \sum_{\Delta=0}^{\tau-1}C(\Delta,0)\pi_\Delta + \sum_{\Delta=\tau}^{\infty}C(\Delta,1)\pi_\Delta,
\end{equation}
where $C(\Delta,a)$ is the expected sum of AoII during the transmission of the update caused by the operation of $a$ at state $(\Delta,0,-1)$. Note that $C(\Delta,a)$ includes the AoII for being at state $(\Delta,0,-1)$.
\begin{remark}\label{rem-CandP}
In order to have a more intuitive understanding of the definition of  $C(\Delta,a)$, we use $\eta$ to denote a possible path of the state during the transmission of the update and let $H$ be the set of all possible paths. Moreover, we denote by $C_{\eta}$ and $P_{\eta}$ the sum of AoII and the probability associated with path $\eta$, respectively. Then,
\[
C(\Delta,a) = \sum_{\eta\in H}P_{\eta}C_{\eta}.
\]
For example, we consider the case of $p_2=1$, where the transmission takes $2$ time slots to be delivered. Also, action $a=1$ is taken at state $(2,0,-1)$. Then, a sample path $\eta$ of the state during the transmission can be the following.
\[
(2,0,-1)\rightarrow(3,1,1)\rightarrow(4,0,-1).
\]
By our definition, $C_{\eta}=2+3=5$ and $P_{\eta} = Pr[(3,1,1)\mid (2,0,-1),a=1]\cdot Pr[(4,0,-1)\mid (3,1,1),a=1]$ for the above sample path.
\end{remark}
In the following, we calculate $C(\Delta,a)$. Similar to Section~\ref{sec-StateTransProb}, we define $C^{t}(\Delta,a)$ as the expected sum of AoII during the transmission of the update caused by action $a$ at state $(\Delta,0,-1)$, given that the transmission takes $t$ time slots. Then, under \textbf{Assumption 1},
\begin{equation}\label{eq-CompactCostAss1}
C(\Delta,a) = \begin{dcases}
\Delta & a=0,\\
\sum_{t=1}^{t_{max}}p_tC^{t}(\Delta,1) & a=1,
\end{dcases}
\end{equation}
and, under \textbf{Assumption 2},
\begin{equation}\label{eq-CompactCostAss2}
C(\Delta,a) = \begin{dcases}
\Delta & a=0,\\
\sum_{t=1}^{t_{max}}p_tC^{t}(\Delta,1) + p_{t^+}C^{t_{max}}(\Delta,1) & a=1.
\end{dcases}
\end{equation}
Hence, obtaining the expressions of $C^{t}(\Delta,1)$ is sufficient. To this end, we define $C^k(\Delta)$ as the expected AoII $k$ time slots after the transmission starts at state $(\Delta,0,-1)$, given that the transmission is still in progress. Then, we have the following lemma.
\begin{lemma}\label{lem-CompactCost}
$C^{t}(\Delta,1)$ is given by
\[
C^{t}(\Delta,1) = \sum_{k=0}^{t-1}C^k(\Delta),
\]
where $C^k(\Delta)$ is given by \eqref{eq-EquivalentEq3}.
\begin{figure*}[!t]
\normalsize
\begin{equation}\label{eq-EquivalentEq3}
C^k(\Delta) = \begin{dcases}
\sum_{h=1}^{k} hp^{(k-h)}p(1-p)^{h-1} & \Delta=0,\\
\sum_{h=1}^{k-1} h(1-p^{(k-h)})p(1-p)^{h-1} + (\Delta+k)(1-p)^k & \Delta>0.
\end{dcases}
\end{equation}
\hrulefill
\vspace*{4pt}
\end{figure*}
\end{lemma}
\begin{proof}
The expression of $C^k(\Delta)$ is obtained by analyzing the system dynamics. The complete proof can be found in Appendix~\ref{pf-ExpectedCost} of the supplementary material.
\end{proof}
Next, we calculate the expected AoII achieved by the threshold policy. We start with the case of $\tau=\infty$.
\begin{theorem}\label{prop-LazyPerformance}
The expected AoII achieved by the threshold policy with $\tau=\infty$ is
\[
\bar{\Delta}_\infty = \frac{1}{2p}.
\]
\end{theorem}
\begin{proof}
In this case, the transmitter will never initiate any transmissions. Hence, the state transitions are straightforward. The complete proof can be found in Appendix~\ref{pf-LazyPerformance} of the supplementary material.
\end{proof}
In the following, we focus on the case where $\tau$ is finite. We recall that the expected AoII is given by \eqref{eq-ExpectedAoII}. The problem arises because of the infinite sum. To overcome this, we adopt a similar approach as proposed in Section~\ref{sec-SteadyState}. More precisely, we leverage the structural property of the threshold policy and define $\Sigma\triangleq \sum_{\Delta=\omega}^{\infty}C(\Delta,1)\pi_\Delta$. Then, equation \eqref{eq-ExpectedAoII} can be written as
\[
\bar{\Delta}_{\tau} = \sum_{i=0}^{\tau-1}C(i,0)\pi_i + \sum_{i=\tau}^{\omega-1}C(i,1)\pi_i + \Sigma.
\]
As we have obtained the expressions of $\pi_\Delta$ and $C(\Delta,a)$ in previous subsections, it is sufficient to obtain the expression of $\Sigma$.
\begin{theorem}\label{thm-Sigma}
Under \textbf{Assumption 1} and for $0\leq\tau<\infty$,
\[
\Sigma = \ddfrac{\sum_{t=1}^{t_{max}}\left[p_tP^t_{1,1+t}(1)\left(\sum_{i=\omega-t}^{\omega-1}C(i,1)\pi_i\right) + \Delta_t'\Pi_t\right]}{1-\sum_{t=1}^{t_{max}}\bigg(p_tP^t_{1,1+t}(1)\bigg)},
\]
where
\[
\Pi_t= p_{t}P^{t}_{1,1+t}(1)\left(\sum_{i=\omega-t}^{\omega-1}\pi_i + \Pi\right),
\]
\[
\Delta_t' = \sum_{i=1}^{t_{max}}p_i\left(\frac{t-t(1-p)^i}{p}\right).
\]
Under \textbf{Assumption 2} and for $0\leq\tau<\infty$,
\[
\Sigma = \ddfrac{\sum_{t=1}^{t_{max}}\left[\left(\sum_{i=\omega-t}^{\omega-1}\Upsilon(i+t,t)C(i,1)\pi_i\right) + \Delta_t'\Pi_t\right]}{1-\sum_{t=1}^{t_{max}}\Upsilon(\omega+t,t)},
\]
where
\[
\Upsilon(\Delta,t)= p_tP^t_{\Delta-t,\Delta}(1) + p_{t^+}P^{t^+}_{\Delta-t,\Delta}(1),
\]
\[
\Pi_t = \sum_{i=\omega-t}^{\omega-1}\Upsilon(i+t,t)\pi_i + \Upsilon(\omega+t,t)\Pi,
\]
\[
\Delta_t' = \sum_{i=1}^{t_{max}}p_i\left(\frac{t-t(1-p)^i}{p}\right) + p_{t^+}\left(\frac{t-t(1-p)^{t_{max}}}{p}\right).
\]
\end{theorem}
\begin{proof}
We delve into the definition of $\Sigma$ and repeatedly use the properties of $C(\Delta,a)$ and $P_{\Delta,\Delta'}(a)$. The complete proof can be found in Appendix~\ref{pf-Performance} of the supplementary material.
\end{proof}

\section{Optimal Policy}\label{sec-OptimalPolicy}
In this section, we find the optimal policy for $\mathcal{M}$ theoretically. First of all, we prove that the optimal policy exists. 
\subsection{Existence of  Optimal Policy}\label{sec-OptimalPolicyExistence}
We introduce the infinite horizon $\gamma$-discounted cost of $\mathcal{M}$, where $0 < \gamma < 1$ is a discount factor. The expected $\gamma$-discounted cost under policy $\phi$ is
\begin{equation}\label{eq-GammaDiscount}
V_{\phi,\gamma}(s) = \mathbb{E}_{\phi}\left[\sum_{t=0}^{\infty}\gamma^tC(s_t)\mid s_0\right],
\end{equation}
where $s_t$ is the state of $\mathcal{M}$ at time slot $t$. We define $V_{\gamma}(s) \triangleq \inf_{\phi} V_{\phi,\gamma}(s)$ as the best that can be achieved. Equivalently, $V_{\gamma}(s)$ is  the value function associated with the $\gamma$-discounted version of $\mathcal{M}$. Hence, $V_{\gamma}(s)$ satisfies the corresponding Bellman equation.
\[
V_{\gamma}(s) = \min_{a\in\mathcal{A}}\left\lbrace C(s)+\gamma\sum_{s'\in\mathcal{S}}P_{s,s'}(a)V_{\gamma}(s')\right\rbrace.
\]
Value iteration algorithm is a canonical algorithm to calculate $V_{\gamma}(s)$. Let $V_{\gamma,\nu}(s)$ be the estimated value function at iteration $\nu$. Then, the estimated value function is updated in the following way.
\begin{equation}\label{eq-ValueIterationGammaDiscount}
V_{\gamma,\nu+1}(s) = \min_{a\in\mathcal{A}}\left\lbrace C(s)+\gamma\sum_{s'\in\mathcal{S}}P_{s,s'}(a)V_{\gamma,\nu}(s')\right\rbrace.
\end{equation}
\begin{lemma}\label{lem-Converge}
When updated following \eqref{eq-ValueIterationGammaDiscount}, $\lim_{\nu\rightarrow\infty}V_{\gamma,\nu}(s)= V_{\gamma}(s)$.
\end{lemma}
\begin{proof}
According to~\cite[Propositions 1 and 3]{b17}, it is sufficient to show that $V_{\gamma}(s)$ is finite. To this end, we consider the policy $\phi$ being the threshold policy with $\tau=\infty$. According to \eqref{eq-GammaDiscount}, we have
\begin{align*}
V_{\phi,\gamma}(s) & = \mathbb{E}_{\phi}\left[\sum_{t=0}^{\infty}\gamma^tC(s_t)\ |\ s_0\right]\\
& \leq \sum_{t = 0}^{\infty}\gamma^t (\Delta_0+t) = \frac{\Delta_0}{1-\gamma} + \frac{\gamma}{(1-\gamma)^2}\\
& < \infty.
\end{align*}
Then, by definition, we have $V_{\gamma}(s)\leq V_{\gamma,\phi}(s)<\infty$. Hence, the value iteration reported in \eqref{eq-ValueIterationGammaDiscount} will converge to the value function.
\end{proof}
Leveraging the convergence of the value iteration algorithm, we can prove the following structural property of $V_{\gamma}(s)$.
\begin{lemma}\label{lem-Monotone}
$V_{\gamma}(s)$ is non-decreasing in $\Delta$ when $\Delta>0$.
\end{lemma}
\begin{proof}
We recall that $V_{\gamma}(s)$ can be calculated using the value iteration algorithm. Hence, the monotonicity of $V_{\gamma}(s)$ can be proved via mathematical induction. The complete proof can be found in Appendix~\ref{pf-Monotone} of the supplementary material.
\end{proof}
Now, we proceed with showing the existence of the optimal policy. To this end, we first define the stationary policy.
\begin{definition}[Stationary policy]
A stationary policy specifies a single action in each time slot.
\end{definition}
\begin{theorem}\label{thm-optimalexist}
There exists a stationary policy that is optimal for $\mathcal{M}$. Moreover, the minimum expected AoII is independent of the initial state.
\end{theorem}
\begin{proof}
We show that $\mathcal{M}$ verifies the two conditions given in~\cite{b17}. Then, the results in the theorem is guaranteed by~\cite[Theorem]{b17}. The complete proof can be found in Appendix~\ref{pf-optimalexist} of the supplementary material.
\end{proof}
We denote by $\phi^*$ the optimal policy for $\mathcal{M}$. Then, the next problem is how to find $\phi^*$. To solve MDP, the value iteration algorithm and the policy iteration algorithm are two of the most popular. In the value iteration algorithm, the value function $V(s)$ is computed iteratively until convergence. However, since the state space $\mathcal{S}$ is infinite, it is not feasible to compute the value function for all states. To make the calculation feasible, in Section~\ref{sec-VIA}, an approximation method is used to obtain an approximated optimal policy $\hat{\phi}^*$, and we rigorously prove that $\hat{\phi}^*$ converges to $\phi^*$. However, the choice of the approximation parameters can significantly affect the complexity of the algorithm and may even lead to a non-optimal policy. To avoid this problem, in Section~\ref{sec-PIA}, we introduce the policy iteration algorithm and find $\phi^*$ theoretically using the policy improvement theorem. We start with the value iteration algorithm in the following subsection.

\subsection{Value Iteration Algorithm}\label{sec-VIA}
In this subsection, we present the relative value iteration (RVI) algorithm that approximates $\phi^*$. Direct application of RVI becomes impractical as the state space $\mathcal{S}$ is infinite. Hence, we use approximating sequence method (ASM)~\cite{b23}. To this end, we construct another MDP $\mathcal{M}^{(m)} = (\mathcal{S}^{(m)},\mathcal{A},\mathcal{P}^{(m)},\mathcal{C})$ by truncating the value of $\Delta$. More precisely, we impose
\[
\mathcal{S}^{(m)}:\begin{dcases}
\Delta \in \{0,1,...,m\}, &\\
i \in \{-1,0,1\}, &\\
t \in \{0,1,...,t_{max}-1\}, &
\end{dcases}
\]
where $m$ is the predetermined maximal value of $\Delta$. The transition probabilities from $s\in\mathcal{S}^{(m)}$ to $z\in\mathcal{S}\setminus\mathcal{S}^{(m)}$ are redistributed to the states $s'\in\mathcal{S}^{(m)}$ in the following way.
\[
P^{(m)}_{s,s'}(a) = \begin{dcases}
          P_{s,s'}(a) & \Delta'<m, \\
          P_{s,s'}(a) + \sum_{G(z,s')}P_{s,z}(a) & \Delta'=m, 
     \end{dcases}
\]
where $G(z,s') = \{z = (\Delta,t,i):\Delta>m,t=t',i=i'\}$. The action space $\mathcal{A}$ and the immediate cost $\mathcal{C}$ are the same as defined in $\mathcal{M}$.
\begin{theorem}\label{thm-ASMConvergence}
The sequence of optimal policies for $\mathcal{M}^{(m)}$ will converge to the optimal policy for $\mathcal{M}$ as $m\rightarrow\infty$.
\end{theorem}
\begin{proof}
The proof follows the same steps as those in the proof of~\cite[Theorem 1]{b8}. The complete proof can be found in Appendix~\ref{pf-ASMConvergence} of the supplementary material.
\end{proof}
Then, we can apply RVI to $\mathcal{M}^{(m)}$ and treat the resulting policy as an approximation of $\phi^*$. The pseudocode of RVI is given in Algorithm~\ref{alg-RVIA}. 
\begin{algorithm}[!t]
    \begin{algorithmic}[1]
    \Procedure{RVI}{$\mathcal{M}^{(m)}$,$\epsilon$}
        \State $V_0(s)\gets0$ for $s\in\mathcal{S}^{(m)}$; $\nu\gets0$
        \State Choose $s^{ref}\in\mathcal{S}^{(m)}$ arbitrarily
        \Repeat
            \For{$s \in \mathcal{S}^{(m)}$}
                \For{$a \in \mathcal{A}$}
                    \State $H_{s,a} \gets C(s) + \sum_{s'} P^{(m)}_{s,s'}(a) V_{\nu}(s')$
                \EndFor
            	\State $Q_{\nu+1}(s) \gets \min_a \{H_{s,a}\}$
            \EndFor
            \For{$s \in \mathcal{S}^{(m)}$}
                \State $V_{\nu+1}(s) \gets Q_{\nu+1}(s) - Q_{\nu+1}(s^{ref})$
            \EndFor
            \State $\nu \gets \nu + 1$
        \Until{$\max_s\{\left|V_{\nu}(s)-V_{\nu-1}(s)\right|\}\leq\epsilon$}
        \State \Return $\hat{\phi}^* \gets \argmin_a \{H_{s,a}\}$
    \EndProcedure
    \end{algorithmic}
\caption{Relative Value Iteration}
\label{alg-RVIA}
\end{algorithm}
However, the choice of the approximation parameter $m$ is crucial. A large $m$ can add unnecessary computational complexity, while a small $m$ can lead to a non-optimal policy. Therefore, in the following subsections, we use the policy iteration algorithm and the policy improvement theorem to find $\phi^*$ theoretically. We start with introducing the policy iteration algorithm.

\subsection{Policy Iteration Algorithm}\label{sec-PIA}
The policy iteration algorithm is an iterative algorithm that iterates between the following two steps until convergence.\footnote{The convergence happens when two consecutive iterations produce equivalent policies.}
\begin{enumerate}
\item The first step is policy evaluation. In this step, we calculate the value function $V^{\phi}(s)$ and the expected AoII $\theta^{\phi}$ resulting from the adoption of some policy $\phi$. More precisely, the value function and the expected AoII are obtained by solving the following system of linear equations.
\begin{equation}\label{eq-policyevaluate}
V^{\phi}(s) + \theta^{\phi} = C(s) + \sum_{s'\in\mathcal{S}}P^{\phi}_{s,s'}V^{\phi}(s')\quad s\in\mathcal{S},
\end{equation}
where $P^{\phi}_{s,s'}$ is the state transition probability from $s$ to $s'$ when policy $\phi$ is adopted. Note that \eqref{eq-policyevaluate} forms an underdetermined system. Hence, we can select a reference state $s$ arbitrarily and set the corresponding value function to $0$. In this way, we can obtain a unique solution.
\item The second step is policy improvement. In this step, we obtain a new policy $\phi'$ using the $V^{\phi}(s)$ obtained in the first step. More precisely, the action suggested by $\phi'$ at state $s$ is determined by
\[
\phi'(s) = \argmin_{a\in\mathcal{A}}\left\lbrace	C(s) + \sum_{s'\in\mathcal{S}}P_{s,s'}(a)V^{\phi}(s')\right\rbrace.
\]
\end{enumerate}
The pseudocode of the policy iteration algorithm is given in Algorithm~\ref{alg-PIA}.
\begin{algorithm}[!t]
    \begin{algorithmic}[1]
    \Procedure{PI}{$\mathcal{M}$}
        \State Choose $\phi'(s)\in\mathcal{A}$ arbitrarily for all $s\in\mathcal{S}$
        \Repeat
        \State $\phi(s)\gets\phi'(s)$ for all $s\in\mathcal{S}$
        \State $(V^{\phi}(s),\theta^{\phi})\leftarrow$ \textsc{PolicyEvaluation}$(\mathcal{M},\phi(s))$
        \State $\phi'(s)\leftarrow$ \textsc{PolicyImprovement}$(\mathcal{M},V^{\phi}(s))$
       	\Until{$\phi'(s)=\phi(s)$ for all $s\in\mathcal{S}$}
        \State\Return $(\phi^*,\theta)\gets (\phi(s),\theta^{\phi})$
    \EndProcedure
    \end{algorithmic}
\caption{Policy Iteration}
\label{alg-PIA}
\end{algorithm}
With policy iteration algorithm in mind, we can proceed with presenting the policy improvement theorem.
\begin{theorem}[Policy improvement theorem]\label{thm-PIT}
Suppose that we have obtained the value function resulting from the operation of a policy $A$ and that the subsequent policy improvement step has produced a policy $B$, the following results hold.
\begin{itemize}
\item If $B$ is different from $A$, $\theta^A\geq \theta^B$.
\item If $A$ and $B$ are equivalent,\footnote{Policies $A$ and $B$ are equivalent when they yield the same expected AoII.} both policies are optimal.
\end{itemize}
\end{theorem}
\begin{proof}
The proof is based on~\cite[pp.~42-43]{b18}. The complete proof can be found in Appendix~\ref{pf-PIT} of the supplementary material.
\end{proof}
With the most important theorem proved, we proceed with finding $\phi^*$ theoretically. First, we simplify the Bellman equation shown in \eqref{eq-Bellman} to make the theoretical proof more concise and straightforward.

\subsection{Simplifying the Bellman Equation}
We note that state transitions are complex and intertwined. Consequently, the direct analysis of the Bellman equation \eqref{eq-Bellman} is complicated. In the following, we will simplify the Bellman equation. To this end, we leverage the fact that the action space depends on the state space. More specifically, when the channel is busy (i.e., $i\neq-1$), the feasible action is $a=0$. Hence, the transmitter's actions at these states are fixed, which leads to the fact that for these states, the minimum operators in \eqref{eq-Bellman} are avoided. Let $\mathcal{S}_{-1}\triangleq\{s:i=-1\}$ be the set of states at which the channel is idle. Then,
\begin{equation}\label{eq-BellmanOneAction}
\begin{split}
V(s) + \theta = & \min_{a\in\mathcal{A}}\left\lbrace C(s) + \sum_{s'\in\mathcal{S}}P_{s,s'}(a)V(s')\right\rbrace\\
= & C(s) + \sum_{s'\in\mathcal{S}}P_{s,s'}(0)V(s')\quad s\in\mathcal{S}\setminus\mathcal{S}_{-1}.
\end{split}
\end{equation}
Then, by repeatedly replacing the $V(s)$, where $s\in\mathcal{S}\setminus\mathcal{S}_{-1}$, with the expression given by \eqref{eq-BellmanOneAction}, we can obtain the Bellman equation consists only $V(s)$ where $s\in\mathcal{S}_{-1}$. We know that $s=(\Delta,0,-1)$ for $s\in\mathcal{S}_{-1}$. Hence, we abbreviate $V(\Delta,0,-1)$ as $V(\Delta)$. Then, for each $\Delta\geq0$, we have the following modified Bellman equation.
\begin{multline}\label{eq-CompactBellman}
V(\Delta)+\theta = \\
\min_{a\in\{0,1\}}\left\lbrace C(\Delta,a)  - \theta(a) + \sum_{\Delta'\geq0}P_{\Delta,\Delta'}(a)V(\Delta')\right\rbrace,
\end{multline}
where
\[
\theta(a)= \begin{dcases}
0 & a=0,\\
(ET-1)\theta & a=1.
\end{dcases}
\]
Note that $ET$, $P_{\Delta',\Delta}(a)$, and $C(\Delta,a)$ are those defined and discussion in Section~\ref{sec-PolicyPerformance}. Hence, it is sufficient to use \eqref{eq-CompactBellman} instead of \eqref{eq-Bellman} to determine the optimal action at state $(\Delta,0,-1)$. Although equation \eqref{eq-CompactBellman} may seem complicated at first glance, its advantages will be fully demonstrated in the following subsection.

\subsection{Optimality Proof}
In this subsection, we find $\phi^*$ theoretically. We first introduce the condition that is essential to the analysis later on.
\begin{condition}\label{con}
The condition is the following.
\[
\bar{\Delta}_1\leq\min\left\lbrace\bar{\Delta}_0, \frac{1+(1-p)\sigma}{2}\right\rbrace,
\]
where, for \textbf{Assumption 1},
\[
\sigma = \ddfrac{\sum_{t=1}^{t_{max}}p_t\left(\frac{1-(1-p)^{t}}{p}\right)}{1-\sum_{t=1}^{t_{max}}pp_t(1-p)^{t-1}},
\]
and for \textbf{Assumption 2},
\[
\sigma = \ddfrac{\sum_{t=1}^{t_{max}}p_t\left(\frac{1-(1-p)^t}{p}\right) + p_{t^+}\left(\frac{1-(1-p)^{t_{max}}}{p}\right)}{1-\left(\sum_{t=1}^{t_{max}}pp_t(1-p)^{t-1} + p_{t^+}(1-p)^{t_{max}}\right)}.
\]
$\bar{\Delta}_0$ and $\bar{\Delta}_1$ are the expected AoII resulting from the adoption of the threshold policy with $\tau=0$ and $\tau=1$, respectively.
\end{condition}
\begin{theorem}\label{thm-optimalpolicy}
When Condition~\ref{con} is satisfied, the optimal policy for $\mathcal{M}$ is the threshold policy with $\tau=1$.
\end{theorem}
\begin{proof}
The value iteration algorithm detailed in Section~\ref{sec-VIA} provides a good guess on the optimal policy. Then, we theoretically prove the optimality using the policy improvement theorem. The general procedure for the optimality proof can be summarized as follows.
\begin{enumerate}
\item \textit{Policy Evaluation}: We calculate the value function resulting from the adoption of the threshold policy with $\tau=1$.
\item \textit{Policy Improvement}: We obtain a new policy using the value function obtained in the previous step and verify that the new policy is the threshold policy with $\tau=1$.
\end{enumerate}
Then, the policy improvement theorem tells us that the threshold policy with $\tau=1$ is optimal. The complete proof can be found in Appendix~\ref{pf-optimalpolicy} of the supplementary material.
\end{proof}
\begin{remark}
When the system fails to satisfy Condition~\ref{con}, we can use the relative value iteration algorithm introduced in Section~\ref{sec-VIA} to obtain a good estimate of $\phi^*$.
\end{remark}

\section{Numerical Results}\label{sec-NumericalResults}
In this section, we numerically verify Condition~\ref{con} and analyze the performance of the optimal policy.
\subsection{Verification of Condition~\ref{con}}\label{sec-Verification}
As the closed-form expressions of $\bar{\Delta}_0$ and $\bar{\Delta}_1$ are given in Section~\ref{sec-PolicyPerformance}, the inequality in Condition~\ref{con} is easy to verify. We verify Condition~\ref{con} numerically for the following systems.
\begin{itemize}
\item System adopts \textbf{Assumption 1}/\textbf{Assumption 2} and the transmission delay follows the Geometric distribution with success probability $p_s$. More precisely, $p_t = (1-p_s)^{t-1}p_s$.
\item System adopts \textbf{Assumption 1} and the transmission delay follows the Zipf distribution with constant $a$. More precisely, $p_t = \frac{t^{-a}}{\sum_{i=1}^{t_{max}}i^{-a}},\ 1\leq t\leq t_{max}$.
\item System adopts \textbf{Assumption 1} and $p_t = \frac{1}{2}(\mathbbm{1}\{t=1\} + \mathbbm{1}\{t=t_{max}\})$.
\end{itemize}
For each of the above systems, the parameters take the following values.
\begin{itemize}
\item $0.05\leq p\leq 0.45$ with step size being equal to $0.05$.
\item $2\leq t_{max}\leq 15$ with step size being equal to $1$.
\item $0\leq p_s\leq 0.95$ with step size being equal to $0.05$.
\item $0\leq a\leq 5$ with step size being equal to $0.25$.
\end{itemize}
The numerical results show that all the above systems satisfy Condition~\ref{con}. Then, according to Theorem~\ref{thm-optimalpolicy}, we can conclude that the corresponding optimal policy is the threshold policy with $\tau=1$.
\begin{remark}
The Zipf distribution reduces to the uniform distribution when $a = 0$, and the Geometric transmission delay reduces to a deterministic transmission delay when $p_s = 0$. We ignore the case of $p=0$ because the dynamic source does not change state in this case. Similarly, we are not interested in the case of $p=0.5$ because the state of the dynamic source is independent of the previous state in this case. Also, we exclude the case of $p_s=1$ because, in this case, the transmission time is deterministic and equal to $1$ time slot. The corresponding optimal policies under various system settings are well studied in~\cite{b6,b7,b8,b9,b10}.
\end{remark}

\subsection{Optimal Policy Performance}\label{sec-Performance}
In this subsection, we analyze the performance of the optimal policy. To this end, we consider the system where the transmission delay follows a Geometric distribution with success probability $p_s$. Moreover, we compare the performance of the optimal policy with that of the threshold policies with $\tau=0$ and $\tau=\infty$. All the results are calculated using Section~\ref{sec-PolicyPerformance}.
\paragraph{The effect of $p$} In this case, we fix $t_{max}=5$ and $p_s=0.7$. Then, we vary $p$ and plot the corresponding results in Fig.~\ref{fig-p}. In the figure, to better show the performance of the optimal policy, we only show parts of the results for the threshold policy with $\tau=\infty$. We notice that, as $p$ increases, the expected AoIIs achieved by the threshold policies with $\tau=0$ and $\tau=1$ increase. This is because when $p$ is large, the dynamic source will be inclined to switch between states. Therefore, the state of the dynamic source is more unpredictable, leading to an increase in the achieved expected AoIIs. Meanwhile, the expected AoII achieved by the threshold policy with $\tau=\infty$ decreases as $p$ increases. To explain this, we first recall that, under the threshold policy with $\tau=\infty$, the receiver's estimate does not change. Also, when $p$ is large, the dynamic source will change states frequently. Therefore, the probability of a situation where the receiver's estimate is always incorrect is small, which makes the resulting AoII small. Also, we notice that \textbf{Assumption 1} and \textbf{Assumption 2} lead to almost the same performance. To explain this, we first note that the only difference between \textbf{Assumption 1} and \textbf{Assumption 2} is whether the update is delivered or discarded when the transmission lasts to the $t_{max}$th time slot after the start of the transmission. However, under our choices of $p_s$ and $t_{max}$, the transmission time of an update rarely reaches $t_{max}$ time slots. Even if it reaches $t_{max}$ time slots, delivery or discarding does not significantly impact the performance, as the receiver's estimate can be correct or incorrect regardless of whether the update is delivered. Therefore, \textbf{Assumption 1} and \textbf{Assumption 2} yield almost the same performance.
\begin{figure}[!t]
\centering
\begin{subfigure}{3in}
\centering
\includegraphics[width=3in]{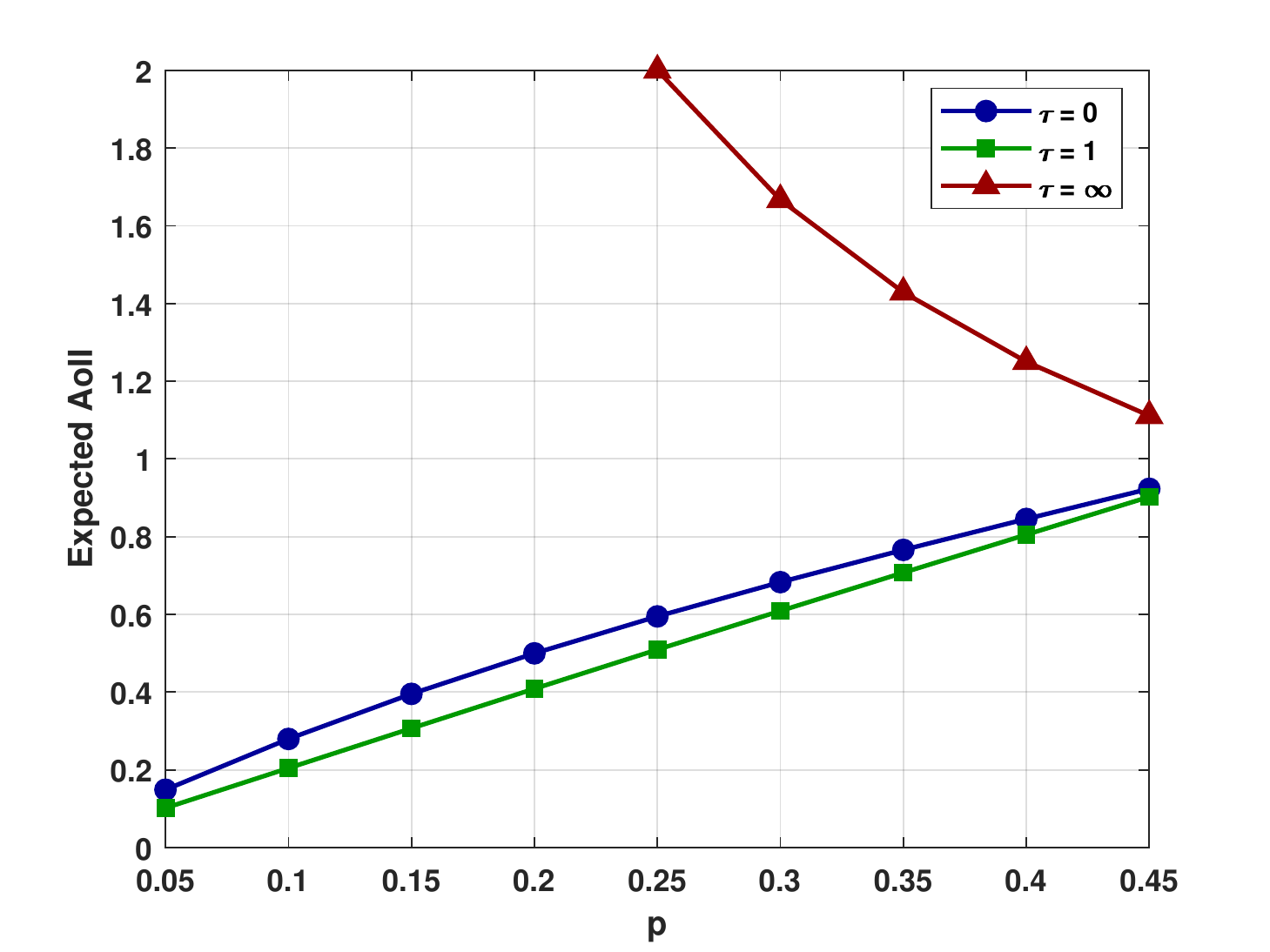}
\caption{Performance under \textbf{Assumption 1}.}
\label{fig-pAss1}
\end{subfigure}\hfill
\begin{subfigure}{3in}
\centering
\includegraphics[width=3in]{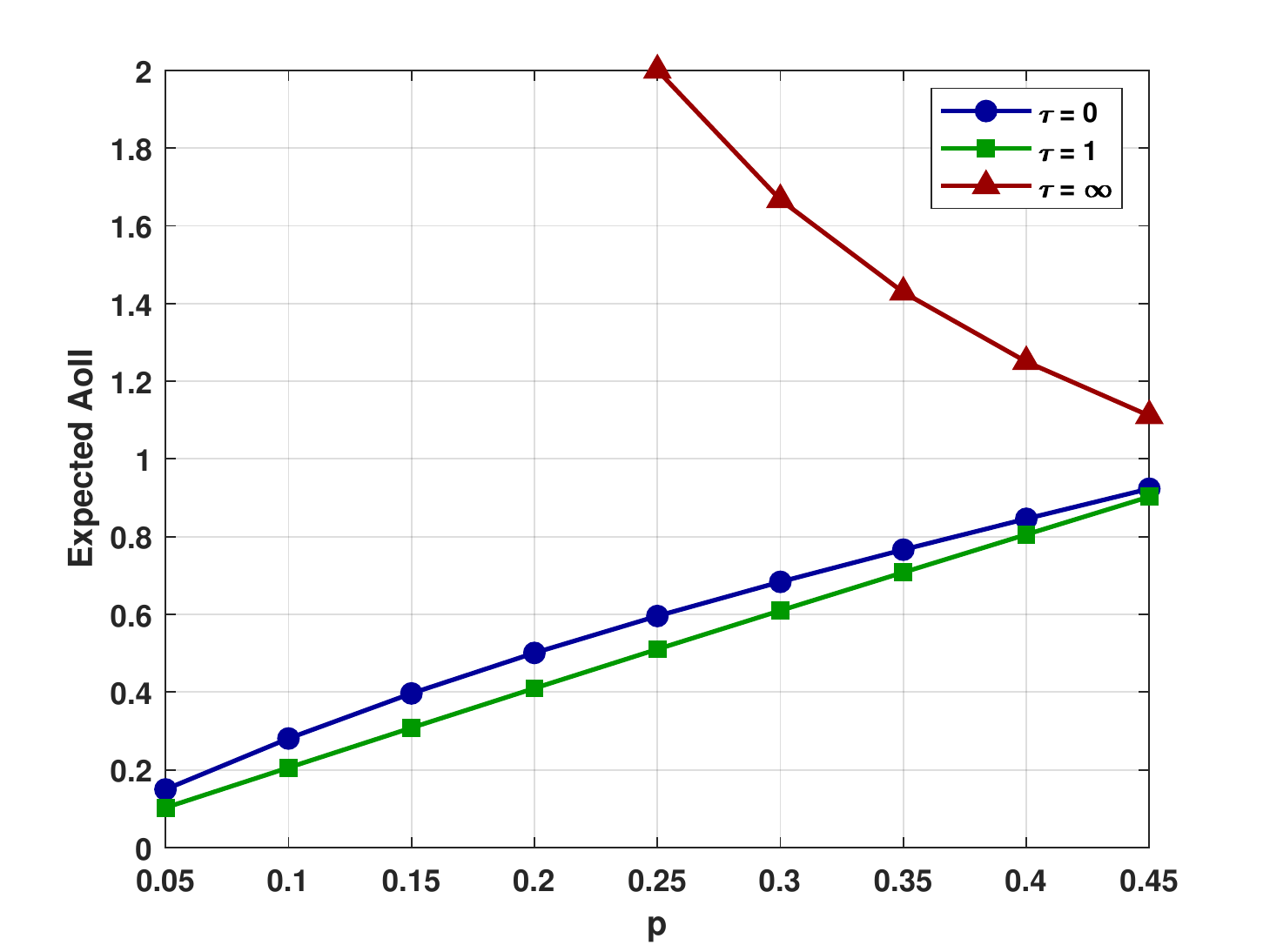}
\caption{Performance under \textbf{Assumption 2}.}
\label{fig-pAss2}
\end{subfigure}\hfill
\caption{Illustrations of the expected AoII as a function of $p$ and $\tau$. We set the upper limit of the transmission time $t_{max}=5$ and the success probability in the Geometric distribution $p_s = 0.7$.}
\label{fig-p}
\end{figure}

\paragraph{The effect of $p_s$} In this case, we fix $t_{max}=5$ and $p=0.35$. Then, we vary $p_s$ and plot the corresponding results in Fig.~\ref{fig-ps}. The figure shows that the expected AoIIs achieved by the threshold policies with $\tau=0$ and $\tau=1$ decrease as $p_s$ increases. The reason behind this is as follows. As $p_s$ increases, the expected transmission time of an update decreases, meaning that updates are more likely to be delivered within the first few time slots. As a result, the receiver receives fresher information, and thus the expected AoII decreases. Moreover, the performance gap between the threshold policies with $\tau=1$ and $\tau=0$ is small when $p_s$ is large. To explain this, we notice that the threshold policy with $\tau=0$ is not optimal because the updates transmitted when AoII is zero do not provide any new information to the receiver. Meanwhile, the transmission will occupy the channel for a few time slots. Therefore, such an action deprives the transmitter of the ability to send new updates for the next few time slots without providing the receiver with any new information. Hence, when $p_s$ is large, the expected transmission time of an update is small. Consequently, the transmission when AoII is zero becomes less costly. Hence, the gap narrows.
\begin{figure}[!t]
\centering
\begin{subfigure}{3in}
\centering
\includegraphics[width=3in]{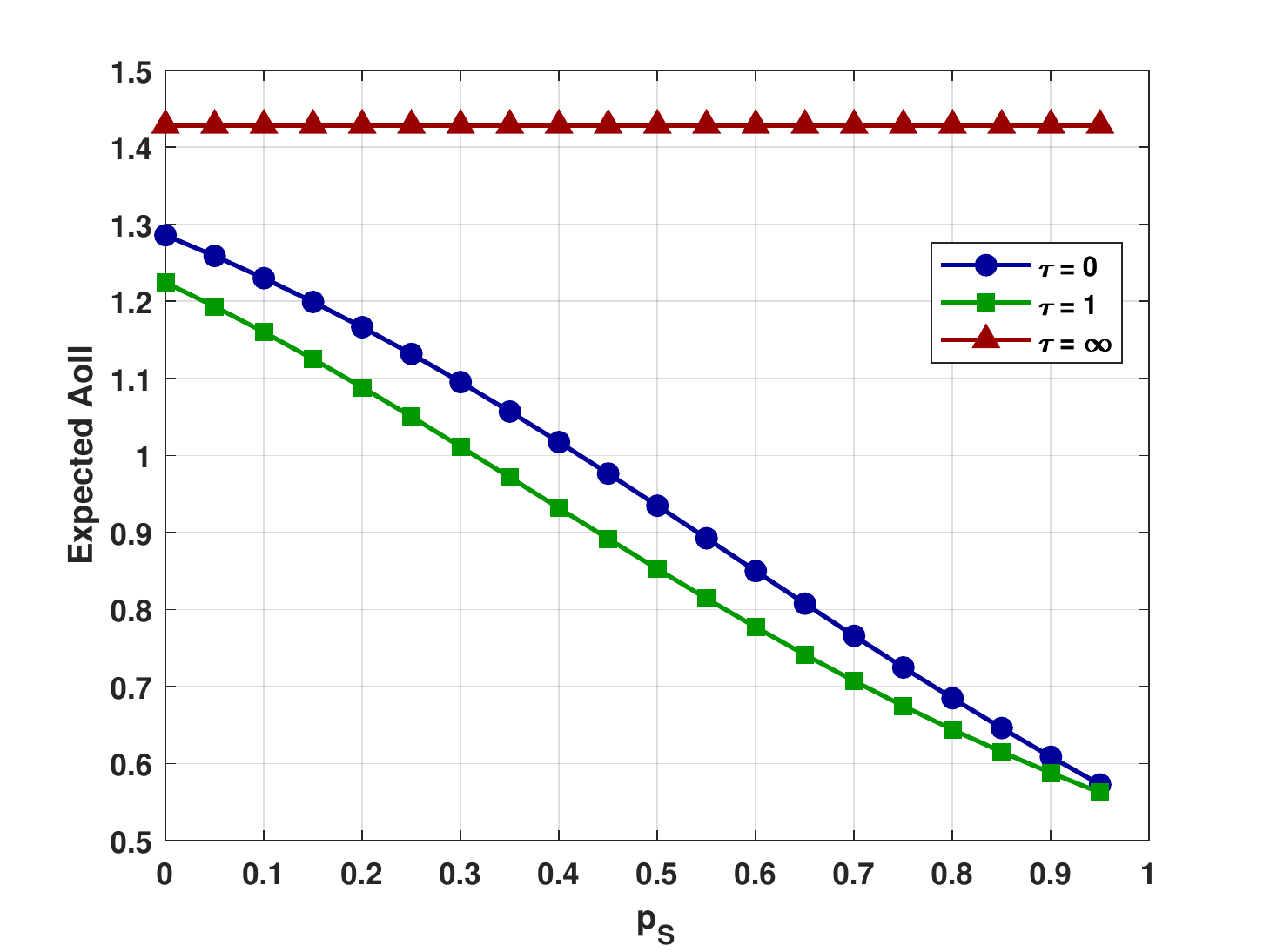}
\caption{Performance under \textbf{Assumption 1}.}
\label{fig-psAss1}
\end{subfigure}\hfill
\begin{subfigure}{3in}
\centering
\includegraphics[width=3in]{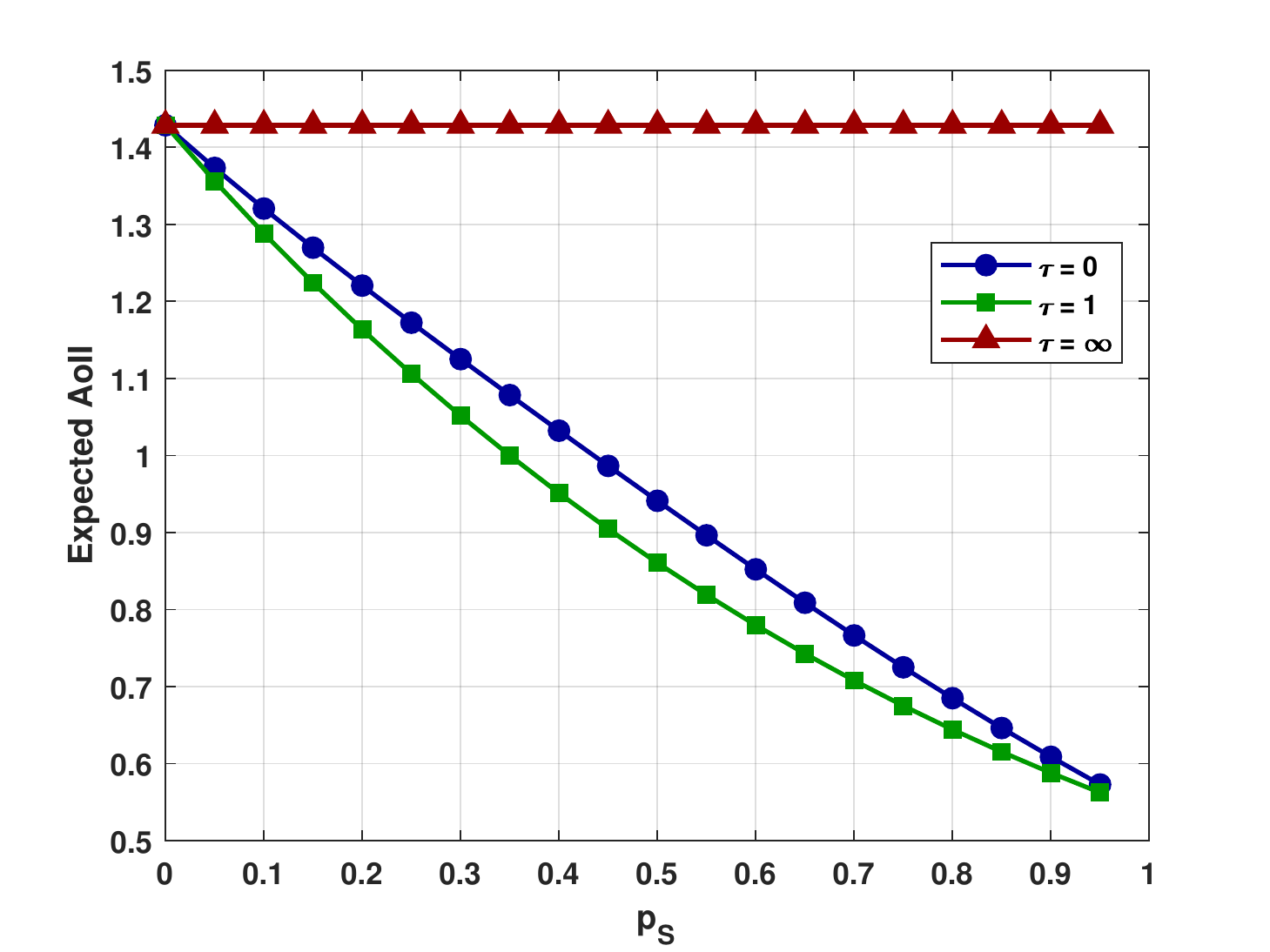}
\caption{Performance under \textbf{Assumption 2}.}
\label{fig-psAss2}
\end{subfigure}\hfill
\caption{Illustrations of the expected AoII as a function of $p_s$ and $\tau$. We set the upper limit of the transmission time $t_{max}=5$ and the source dynamic $p = 0.35$.}
\label{fig-ps}
\end{figure}

\paragraph{The effect of $t_{max}$} In this case, we fix $p_s=0.7$ and $p=0.35$. Then, we vary $t_{max}$ and plot the corresponding results in Fig.~\ref{fig-tmax}. From the figure, we can see that the effect of $t_{max}$ on the performances is only noticeable when $t_{max}$ is small. This is because, under our choice of $p_s$, most updates will be delivered within the first few time slots. Therefore, increasing $t_{max}$ will not significantly affect the performance.
\begin{figure}[!t]
\centering
\begin{subfigure}{3in}
\centering
\includegraphics[width=3in]{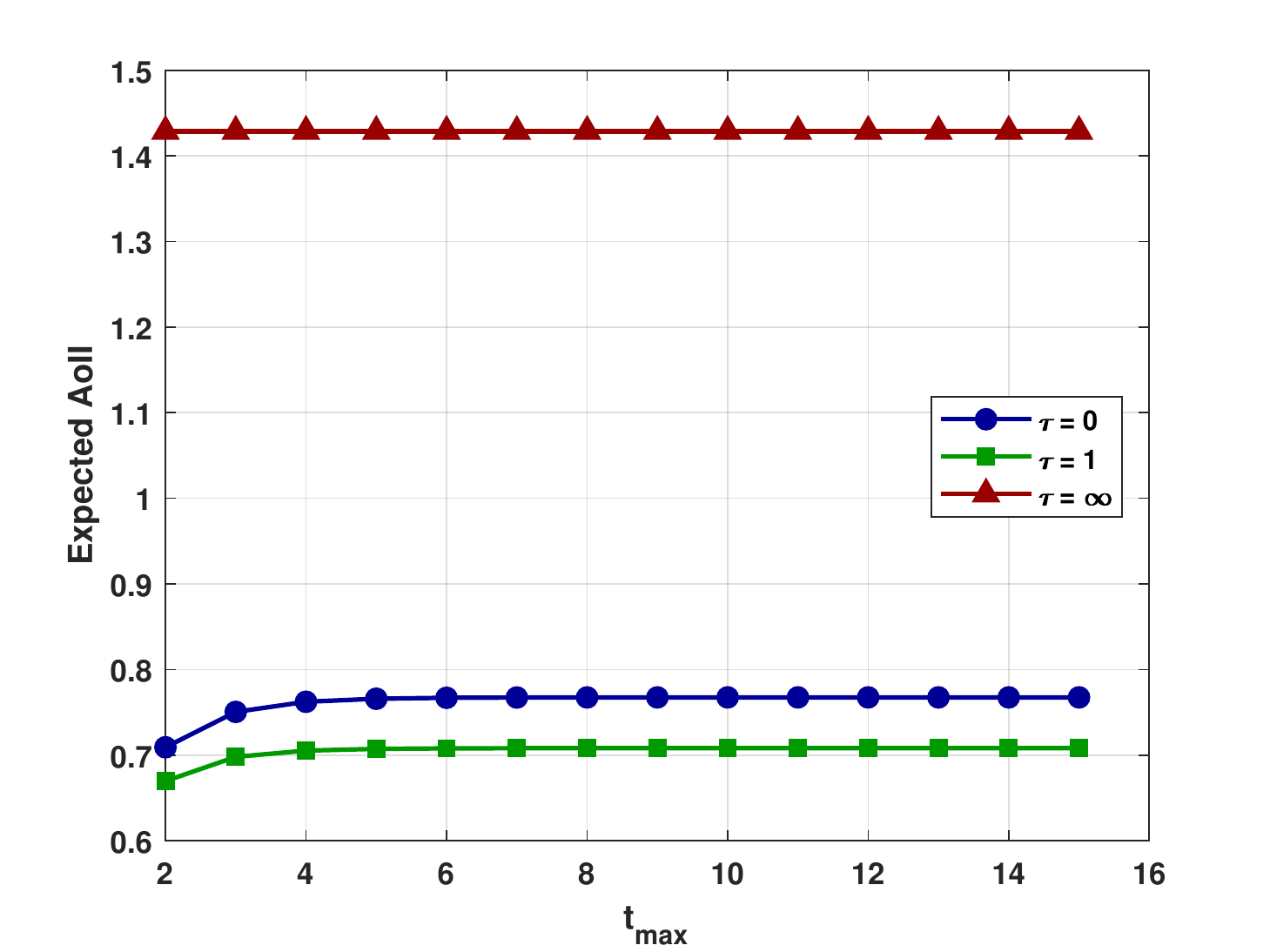}
\caption{Performance under \textbf{Assumption 1}.}
\label{fig-tmaxAss1}
\end{subfigure}\hfill
\begin{subfigure}{3in}
\centering
\includegraphics[width=3in]{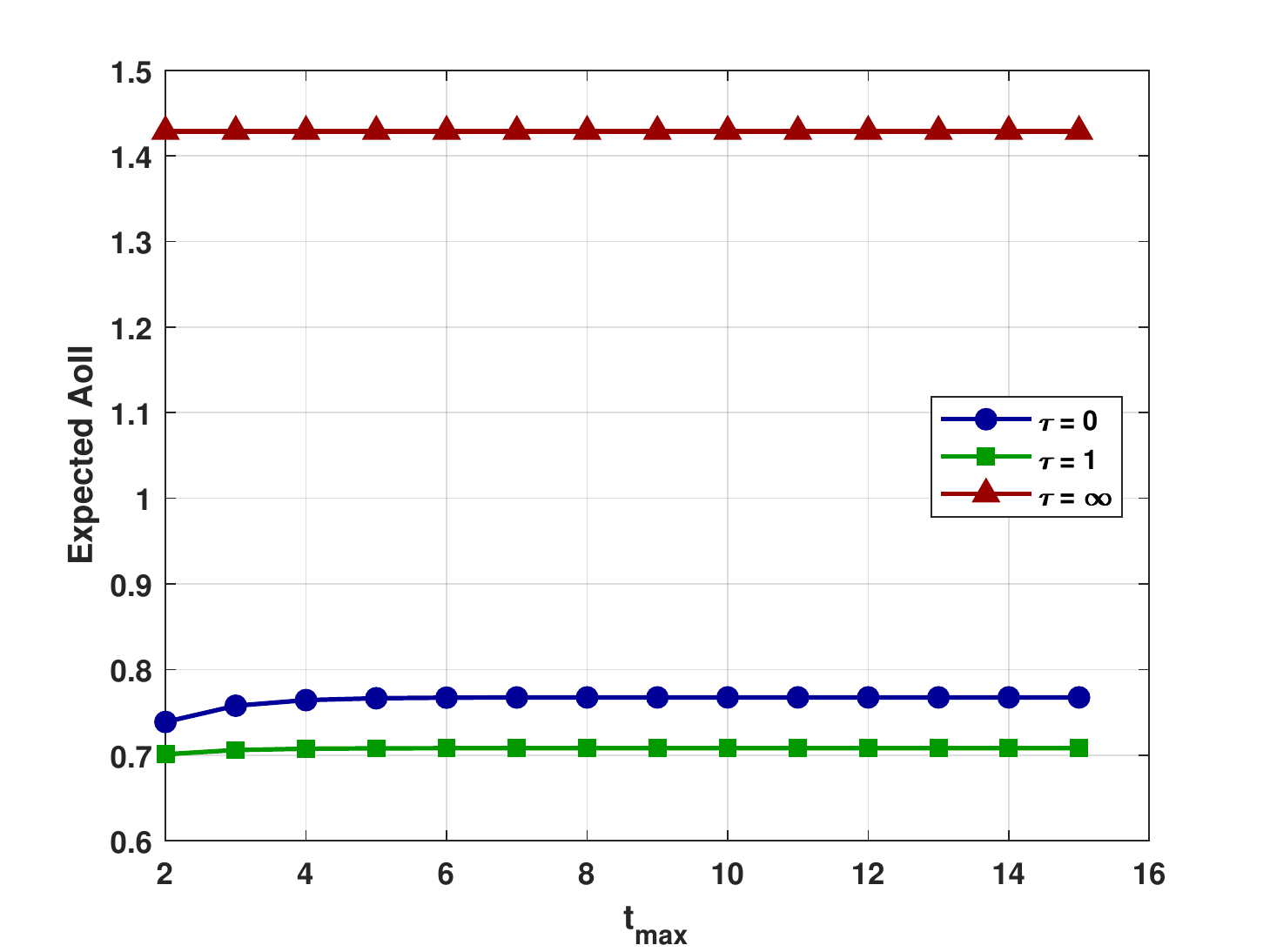}
\caption{Performance under \textbf{Assumption 2}.}
\label{fig-tmaxAss2}
\end{subfigure}\hfill
\caption{Illustrations of the expected AoII as a function of $t_{max}$ and $\tau$. We set the success probability in the Geometric distribution $p_s = 0.7$ and the source dynamic $p = 0.35$.}
\label{fig-tmax}
\end{figure}

\section{Conclusion}
In this paper, we investigate the problem of minimizing the Age of Incorrect Information over a channel with random delay. We study a slotted-time system where a transmitter observes a dynamic source and sends updates to a remote receiver over a channel with random delay. To facilitate the analysis, we consider two cases. The first case assumes that the transmission time has an upper bound and that the update will always be delivered. The second case assumes that the system automatically discards updates if the transmission lasts too long. We aim to find when the transmitter should initiate transmission to minimize the AoII. To this end, we first characterize the optimization problem using the Markov decision process and calculate the expected AoII achieved by the threshold policy precisely using the Markov chain. Next, we prove that the optimal policy exists and provide a computable relative value iteration algorithm to estimate the optimal policy. Then, with the help of the policy improvement theorem, we prove theoretically that, under Condition~\ref{con}, the optimal policy is the threshold policy with $\tau=1$. Finally, we numerically verify Condition~\ref{con} under various system parameters and analyze the performance of the optimal policy.

\bibliographystyle{IEEEtran}
\bibliography{mybib}
\newpage
\setcounter{page}{1}

\twocolumn[
\begin{@twocolumnfalse}
\begin{center}
\Huge Supplementary Material for the Paper "Minimizing Age of Incorrect Information over a Channel with Random Delay"
\end{center}
\end{@twocolumnfalse}
\vspace{3em}]

\appendices
\section{Details of State Transition Probability}\label{app-STP}
We first discuss the individual transition of $\Delta$. We divide our discussion into the following cases.
\begin{itemize}
\item $\Delta=0$ and the receiver's estimates are the same at state $s$ and $s'$. In this case, $\Delta'=0$ when the dynamic source remains in the same state. Otherwise, $\Delta'=1$.
\[
\Delta' = \begin{dcases}
0 & w.p.\ 1-p,\\
1 & w.p.\ p.
\end{dcases}
\]
\item $\Delta=0$ and the receiver's estimates are different at state $s$ and $s'$. In this case, $\Delta'=0$ when the dynamic source flips the state. Otherwise, $\Delta'=1$.
\[
\Delta' = \begin{cases}
0 & w.p.\ p,\\
1 & w.p.\ 1-p.
\end{cases}
\]
\item $\Delta>0$ and the receiver's estimates are the same at state $s$ and $s'$. In this case, $\Delta'=\Delta+1$ when the dynamic source remains in the same state. Otherwise, $\Delta'=0$.
\[
\Delta' = \begin{cases}
0 & w.p.\ p,\\
\Delta'+1 & w.p.\ 1-p.
\end{cases}
\]
\item $\Delta>0$ and the receiver's estimates are different at state $s$ and $s'$. In this case, $\Delta'=\Delta+1$ when the dynamic source flips the state. Otherwise, $\Delta'=0$.
\[
\Delta' = \begin{cases}
0 & w.p.\ 1-p,\\
\Delta+1 & w.p.\ p.
\end{cases}
\]
\end{itemize}
Hence, in the following, we only state whether the receiver's estimates are the same at state $s$ and $s'$ and omit the rest of the discussion on the transition of $\Delta$. To make the notation clearer, we write $P_{s,s'}(a)$ as $P[(\Delta',i',t')\mid (\Delta,t,i),a]$ and $Pr(T>t+1\mid t)$ as $q_{t+1}$ in this proof. Then, we distinguish between the following cases.
\begin{itemize}
\item $s=(0,0,-1)$. In this case, the channel is idle. Hence, the feasible action is $a\in\{0,1\}$. When the transmitter decides not to initiate a new transmission (i.e., $a=0$), $i'=0$ and $t'=-1$. Moreover, the receiver's estimate remains the same. Hence,
\[
Pr[(0,0,-1)\mid (0,0,-1),a=0] = 1-p.
\]
\[
Pr[(1,0,-1)\mid (0,0,-1),a=0] = p.
\]
When the transmitter decides to initiate a new transmission (i.e., $a=1$), the update will be delivered after a random amount of time $T$. When $T>1$, which happens with probability $q_1$, the channel will be busy at the next time slot and $t'=1$ as the transmission starts. Since $\Delta=0$ when the transmission starts, we know $i'=0$. Moreover, the receiver's estimate remains the same since no new update will be delivered. Hence,
\[
Pr[(0,1,0)\mid (0,0,-1),a=1] = q_1(1-p).
\]
\[
Pr[(1,1,0)\mid (0,0,-1),a=1] = q_1p.
\]
When $T=1$, which happens with probability $1-q_1$, the update will be delivered at the next time slot. Hence, the channel will be available for a new transmission at the next time slot, which means that $t'=0$ and $i'=-1$. Since $\Delta=0$ when the transmission starts, the newly arrived update brings no new information to the receiver. Hence, the receiver's estimate remains the same. Hence,
\[
Pr[(0,0,-1)\mid (0,0,-1),a=1] =(1-q_1)(1-p).
\]
\[
Pr[(1,0,-1)\mid (0,0,-1),a=1] =(1-q_1)p.
\]
\item $s = (0,t,0)$. In this case, the channel is busy. Hence, the feasible action is $a=0$. When the update will not arrive at the next time slot, which happens with probability $q_{t+1}$, $i'=i$ since both the transmitting update and the receiver's estimate remain the same. Apparently, $t'=t+1$ as the transmission continues. Moreover, the receiver's estimate remains the same. Hence,
\[
Pr[(0,t+1,0)\mid (0,t,0)] = q_{t+1}(1-p).
\]
\[
Pr[(1,t+1,0)\mid (0,t,0)] = q_{t+1}p.
\]
When the update arrives at the next time slot, which happens with probability $1-q_{t+1}$, $t'=0$ and $i'=-1$ by definition. Since $i=0$, the newly arrived update brings no new information to the receiver. Hence, the receiver's estimate remains the same. Hence,
\[
Pr[(0,0,-1)\mid (0,t,0)] = (1-q_{t+1})(1-p).
\]
\[
Pr[(1,0,-1)\mid (0,t,0)] = (1-q_{t+1})p.
\]
\item $s= (0,t,1)$. The analysis is very similar to that for $s = (0,t,0)$ except that when the update arrives, the receiver's estimate is flipped. Hence,
\[
Pr[(0,t+1,1)\mid (0,t,1)] = q_{t+1}(1-p).
\]
\[
Pr[(1,t+1,1)\mid (0,t,1)] = q_{t+1}p.
\]
\[
Pr[(0,0,-1)\mid (0,t,1)] = (1-q_{t+1})p.
\]
\[
Pr[(1,0,-1)\mid (0,t,1)] = (1-q_{t+1})(1-p).
\]
\item $s= (\Delta,0,-1)$ where $\Delta>0$. In this case, the analysis is very similar to that for $s = (0,0,-1)$, except that the receiver's estimate is incorrect at state $s$, and if the decision is made to transmit, the transmitted update differs from the receiver's estimate. Therefore, the details are omitted here.
\[
Pr[(\Delta+1,0,-1)\mid (\Delta,0,-1),a=0] = 1-p.
\]
\[
Pr[(0,0,-1)\mid (\Delta,0,-1),a=0] = p.
\]
\[
Pr[(\Delta+1,1,1)\mid (\Delta,0,-1),a=1] = q_1(1-p).
\]
\[
Pr[(0,1,1)\mid (\Delta,0,-1),a=1] = q_1p.
\]
\[
Pr[(\Delta+1,0,-1)\mid (\Delta,0,-1),a=1] = (1-q_1)p.
\]
\[
Pr[(0,0,-1)\mid (\Delta,0,-1),a=1] = (1-q_1)(1-p).
\]
\item $s= (\Delta,t,0)$ where $\Delta>0$. The analysis is very similar to that for $s = (0,t,0)$ except that the receiver's estimate is incorrect at state $s$. Hence,
\[
Pr[(\Delta+1,t+1,0)\mid (\Delta,t,0)] = q_{t+1}(1-p).
\]
\[
Pr[(0,t+1,0)\mid (\Delta,t,0)] = q_{t+1}p.
\]
\[
Pr[(\Delta+1,0,-1)\mid (\Delta,t,0)] = (1-q_{t+1})(1-p).
\]
\[
Pr[(0,0,-1)\mid (\Delta,t,0)] = (1-q_{t+1})p.
\]
\item $s=(\Delta,t,1)$ where $\Delta>0$. The analysis is very similar to that for $s = (\Delta,t,0)$ except that the transmitted update differs from the receiver's estimate. Hence,
\[
Pr[(\Delta+1,t+1,1)\mid (\Delta,t,1)] = q_{t+1}(1-p).
\]
\[
Pr[(0,t+1,1)\mid (\Delta,t,1)] = q_{t+1}p.
\]
\[
Pr[(\Delta+1,0,-1)\mid (\Delta,t,1)] = (1-q_{t+1})p.
\]
\[
Pr[(0,0,-1)\mid (\Delta,t,1)] = (1-q_{t+1})(1-p).
\]
\end{itemize}
Combing the above cases, we fully characterized the state transitions and the corresponding probabilities.

\section{Proof of Lemma~\ref{lem-MSTPAss1}}\label{pf-CompactTrans}
We recall that $P^t_{\Delta,\Delta'}(1)$ is the probability that action $a$ at state $s=(\Delta,0,-1)$ will lead to state $s'=(\Delta',0,-1)$, given that the transmission takes $t$ time slots. With this in mind, we first distinguish between different values of $\Delta$.
\begin{itemize}
\item When $\Delta=0$, the transmitted update is the same as the receiver's estimate. Hence, the receiver's estimate will not change due to receiving the transmitted update. Moreover, we recall that AoII will either increases by one or decreases to zero. Hence, $\Delta'\in\{0,1,...,t\}$. Then, we further distinguish our discussion into the following cases.
\begin{itemize}
\item $\Delta'=0$ happens when the receiver's estimate is correct as a result of receiving the update. Hence, the probability of this happening is $p^{(t)}$.
\item $\Delta'=k\in\{1,...,t\}$ happens when the receiver's estimate is correct at $(t-k)$th time slot after the transmission, which happens with probability $p^{(t-k)}$. Then, the estimate remains incorrect for the remainder of the transmission time. This happens when the source first changes state, then remains in the same state throughout the rest of the transmission. Hence, the probability of this happening is $p(1-p)^{k-1}$. Combining together, $\Delta'=k$ happens with probability $p^{(t-k)}p(1-p)^{k-1}$.
\end{itemize}
Combining together, we have
\[
P^{t}_{0,\Delta'}(1) = 
\begin{dcases}
p^{(t)} & \Delta'=0,\\
p^{(t-k)}p(1-p)^{k-1} & 1\leq\Delta'= k\leq t,\\
0 & otherwise.
\end{dcases}
\]
\item When $\Delta>0$, the transmitted update is different from the receiver's estimate. Hence, the receiver's estimate will flip as a result of receiving the transmitted update. Moreover, we know $\Delta'\in\{0,1,...,t-1,\Delta+t\}$. Hence, we further distinguish between the following cases.
\begin{itemize}
\item $\Delta'=0$ happens in the same case as discussed in the case of $\Delta=0$. Hence, the estimate is correct with probability $p^{(t)}$.
\item $\Delta'=1$ happens when the estimate is correct at $(t-1)$th time slot after the transmission, which happens with probability $1-p^{(t-1)}$. Then, the estimate becomes incorrect as a result of receiving the update. Since the estimate flips upon the arrival of the transmitted update, it happens when the source remains in the same state. Hence, the probability of this happening is $1-p$. Combing together, $\Delta'=1$ happens with probability $(1-p^{(t-1)})(1-p)$.
\item $\Delta'=k\in\{2,...,t-1\}$ happens when the estimate is correct at $(t-k)$th time slot after the transmission, which happens with probability $1-p^{(t-k)}$. Then, the estimate remains incorrect for the remainder of the transmission time. This happens when the dynamic source behaves the following way during the remaining transmission time. The dynamic source should first change state, then remain in the same state, and finally, change state again when the update arrives. This happens with probability $p^2(1-p)^{k-2}$. Hence, $\Delta'=k$ happens with probability $(1-p^{(t-k)})p^{2}(1-p)^{k-2}$.
\item $\Delta'=\Delta+t$ happens when the estimate is incorrect throughout the transmission. Since the estimate will flip when the update is received, this happens when the source stays in the same state until the update arrives. Hence, $\Delta'=\Delta+t$ happens with probability $p(1-p)^{t-1}$.
\end{itemize}
Combining together, for $\Delta>0$, we have
\begin{multline*}
P^{t}_{\Delta,\Delta'}(1) = \\
\begin{dcases}
p^{(t)} & \Delta'=0,\\
(1-p^{(t-1)})(1-p) & \Delta'=1,\\
(1-p^{(t-k)})p^{2}(1-p)^{k-2} & 2\leq\Delta'=k\leq t-1,\\
p(1-p)^{t-1} & \Delta'=\Delta+t,\\
0 & otherwise.
\end{dcases}
\end{multline*}
\end{itemize}
By analyzing the above expressions, we can easily conclude that $P^{t}_{\Delta,\Delta'}(1)$ possesses the following properties.
\begin{itemize}
\item $P^{t}_{\Delta,0}(1)$ and $P^{t}_{\Delta,\Delta+t}(1)$ are both independent of $\Delta$.
\item $P^{t}_{\Delta,\Delta'}(1)$ is independent of $\Delta$ when $\Delta>0$ and $0\leq\Delta'\leq t-1$.
\item $P^{t}_{\Delta,\Delta'}(1) = 0$ when $\Delta'>\Delta+t$ or when $t-1<\Delta'<\Delta+t$.
\end{itemize}
Leveraging the above properties, we can prove the second part of the lemma. The equivalent expression can be obtained easily, so the details are omitted. In the following, we focus on proving the properties of $P_{\Delta,\Delta'}(a)$.
\begin{itemize}
\item \textbf{property 1}: When $\Delta'=0$, $P_{\Delta,0}(1) = \sum_{t=1}^{t_{max}}p_tP^t_{\Delta,0}(1)$ for any $\Delta\geq0$. Since $P^{t}_{\Delta,0}(1)$ is independent of $\Delta$, property 1 holds in this case. Then, we consider the case of $1\leq\Delta'\leq t_{max}-1$ and $\Delta\geq\Delta'$. In this case,
\[
P_{\Delta,\Delta'}(1) = \sum_{t=\Delta'}^{t_{max}}p_tP^t_{\Delta,\Delta'}(1),
\]
where $P^t_{\Delta,\Delta'}(1)$ is independent of $\Delta$. Hence, $P_{\Delta,\Delta'}(1)$ is independent of $\Delta$. Combining together, property 1 holds.
\item \textbf{property 2}: We notice that, when $\Delta'\geq t_{max}$,
\[
P_{\Delta,\Delta'}(1) = p_{t'}P^{t'}_{\Delta,\Delta'}(1) = p_{t'}P^{t'}_{\Delta,\Delta+t'}(1).
\]
We recall that $P^{t'}_{\Delta,\Delta+t'}(1)$ is independent of $\Delta$. Then, we can conclude that $P_{\Delta,\Delta'}(1)$ depends only on $t'$. Thus, property 2 holds.
\item \textbf{property 3}: The equivalent expression in corollary indicates that the property holds when $\Delta'>\Delta+t_{max}$. In the case of $t_{max}-1<\Delta'<\Delta+1$, we have
\[
P_{\Delta,\Delta'}(1) = p_{t'}P^{t'}_{\Delta,\Delta'}(1),
\]
where $t'\leq0$. By definition, $P_{\Delta,\Delta'}(1)=0$. Hence, property 3 holds.
\end{itemize}

\section{Proof of Lemma~\ref{lem-MSTPAss2}}\label{pf-Case2TransProb}
The proof is similar to that of Lemma~\ref{lem-MSTPAss1}. We first derive the expressions of $P^t_{\Delta,\Delta'}(1)$ and $P^{t^{+}}_{\Delta,\Delta'}(1)$. To this end, we start with the case of $\Delta=0$. In this case, the transmitted update is the same as the receiver's estimate. With this in mind, we distinguish between different values of $t$.
\begin{itemize}
\item When $1\leq t<t_{max}$, the update is delivered after $t$ time slot. Hence, $\Delta'\in\{0,1,...,t\}$. Then, we further distinguish between different values of $\Delta'$.
\begin{itemize}
\item $\Delta'=0$ in the case where the receiver's estimate is correct when the update is delivered. Hence, $\Delta'=0$ happens with probability $p^{(t)}$.
\item $\Delta' = k \in\{1,2,...,t\}$ when the receiver's estimate is correct at the $(t-k)$th time slots after the transmission occurs. Then, the source flips the state and remains in the same state for the remainder of the transmission. Hence, $\Delta' = k \in\{1,2,...,t\}$ happens with probability $p^{(t-k)}p(1-p)^{k-1}$.
\end{itemize}
\item When $t = t_{max}$, the update either arrives or be discarded. In this case, $\Delta'\in\{0,1,...,t_{max}\}$. We recall that the update is the same as the receiver's estimate. Hence, the receiver's estimate will not change in both cases. Consequently, $P^{t_{max}}_{0,\Delta'}(1)=P^{t^+}_{0,\Delta'}(1)$, which can be obtained by setting the $t$ in the above case to $t_{max}$.
\end{itemize}
Combining together, for each $1\leq t\leq t_{max}$,
\[
P^{t}_{0,\Delta'}(1) = 
\begin{dcases}
p^{(t)} & \Delta'=0,\\
p^{(t-k)}p(1-p)^{k-1} & 1\leq\Delta'= k\leq t,\\
0 & otherwise.
\end{dcases}
\]
\[
P^{t^+}_{0,\Delta'}(1) = P^{t_{max}}_{0,\Delta'}(1).
\]
Then, we consider the case of $\Delta>0$. We notice that, in this case, the receiver's estimate will flip upon receiving the update. Then, we distinguish between different values of $t$.
\begin{itemize}
\item When $1\leq t<t_{max}$, the update is delivered after $t$ time slots, and the receiver's estimate will flip. Hence, $\Delta'\in\{0,1,...,t-1,\Delta+t\}$. Then, we further distinguish between different values of $\Delta'$.
\begin{itemize}
\item $\Delta'=0$ in the case where the receiver's estimate is correct when the update is received. Hence, $\Delta'=0$ happens with probability $p^{(t)}$.
\item $\Delta'=1$ when the receiver's estimate is correct at $(t-1)$th time slot after the transmission starts and becomes incorrect when the update arrives. Hence, $\Delta'=1$ happens with probability $(1-p^{(t-1)})(1-p)$.
\item $\Delta'=k\in\{2,3,...,t-1\}$ when the receiver's estimate is correct at $(t-k)$th time slot after the transmission starts. Then, the source changes state and remains in the same state. Finally, at the time slot when the update arrives, the source flips state again. Hence, $\Delta'=k\in\{2,3,...,t-1\}$ happens with probability $(1-p^{(t-k)})p^{2}(1-p)^{k-2}$.
\item $\Delta'=\Delta+t$ when the estimate is incorrect throughout the transmission. We recall that the receiver's estimate will flip when the update arrives. Hence, $\Delta'=\Delta+t$ when the source remains in the same state until the update arrives, which happens with probability $p(1-p)^{t-1}$.
\end{itemize}
\item When $t=t_{max}$ and the transmitted update is delivered, the receiver's estimate flips. In this case, $\Delta'\in\{0,1,...,t_{max}-1,\Delta+t_{max}\}$. Hence, $P^{t_{max}}_{\Delta,\Delta'}(1)$ can be obtained by setting the $t$ in the above case to $t_{max}$.
\item When $t=t_{max}$ and the transmitted update is discarded, the receiver's estimate remains the same. In this case, $\Delta'\in\{0,1,...,t_{max}-1,\Delta+t_{max}\}$. Then, we further divide our discussion into the following cases.
\begin{itemize}
\item $\Delta'=0$ when the receiver's estimate is correct at the $t_{max}$the time slot after the transmission starts, which happens when the state of the source at the time slot the update is discarded is different from that when the transmission started. Hence, $\Delta'=0$ happens with probability $1- p^{(t_{max})}$.
\item $\Delta' = k \in\{1,2,...,t_{max}-1\}$ when the receiver's estimate is correct at $(t_{max}-k)$th time slot after the transmission starts. Then, the source changes state and remains in the same state for the remainder of the transmission. Hence, $\Delta'=k\in\{1,2,...,t_{max}-1\}$ happens with probability $(1-p^{(t_{max}-k)})p(1-p)^{k-1}$.
\item $\Delta' = \Delta+t_{max}$ when the source remains in the same state throughout the transmission. Combining with the source dynamic, we can conclude that $\Delta' = \Delta+t_{max}$ happens with probability $(1-p)^{t_{max}}$.
\end{itemize}
\end{itemize}
Combining together, for $\Delta>0$ and each $1\leq t\leq t_{max}$,
\begin{multline*}
P^{t}_{\Delta,\Delta'} (1) =\\
\begin{dcases}
p^{(t)} & \Delta'=0,\\
(1-p^{(t-1)})(1-p) & \Delta'=1,\\
(1-p^{(t-k)})p^{2}(1-p)^{k-2} & 2\leq\Delta'=k\leq t-1,\\
p(1-p)^{t-1} & \Delta'=\Delta+t,\\
0 & otherwise.
\end{dcases}
\end{multline*}
\begin{multline*}
P^{t^+}_{\Delta,\Delta'}(1) = \\
\begin{dcases}
1-p^{(t_{max})} & \Delta'=0,\\
(1-p^{(t_{max}-k)})p(1-p)^{k-1} & 1\leq \Delta'= k\leq t_{max}-1,\\
(1-p)^{t_{max}} & \Delta' = \Delta+t_{max},\\
0 & otherwise.
\end{dcases}
\end{multline*}
By analyzing the above expressions, we can easily conclude that $P^t_{\Delta,\Delta'}(1)$ and $P^{t^{+}}_{\Delta,\Delta'}(1)$ possess the following properties.
\begin{itemize}
\item$P^{t}_{\Delta,\Delta+t}(1)$ and $P^{t^+}_{\Delta,\Delta+t_{max}}(1)$ are independent of $\Delta$ when $\Delta>0$.
\item $P^t_{\Delta,\Delta'}(1)$ is independent of $\Delta$ when $\Delta>0$ and $0\leq\Delta'\leq t-1$.
\item $P^t_{\Delta,\Delta'}(1)=0$ when $\Delta>0$ and $t-1<\Delta'< \Delta+t$.
\item $P^{t^+}_{\Delta,\Delta'}(1)$ is independent of $\Delta$ when $\Delta>0$ and $0\leq\Delta'\leq t_{max}-1$.
\item $P^{t^+}_{\Delta,\Delta'}(1)=0$ when $\Delta>0$ and $t_{max}-1<\Delta'< \Delta+t_{max}$.
\end{itemize}
Leveraging the properties above, we proceed with proving the second part of the lemma. The equivalent expression can be obtained easily by analyzing \eqref{eq-MSTPAss2}. Hence, the details are omitted. In the following, we focus on proving the presented properties.
\begin{itemize}
\item \textbf{property 1}: We notice that, when $0 \leq\Delta' \leq t_{max} -1$ and $\Delta\geq \max\{1,\Delta'\}$,
\[
P_{\Delta,\Delta'}(1) = \sum_{t=\Delta'}^{t_{max}}p_tP^t_{\Delta,\Delta'}(1)+ p_{t^+}P^{t^+}_{\Delta,\Delta'}(1).
\]
Then, we divide the discussion into the following two cases.
\begin{itemize}
\item $\Delta\geq\max\{1,\Delta'\}$ indicates that $\Delta>0$ and $\Delta'<\Delta+t_{max}$. Hence, $P^{t^+}_{\Delta,\Delta'}(1)$ is independent of $\Delta$.
\item $\Delta\geq\max\{1,\Delta'\}$ indicates that $\Delta>0$ and $\Delta'<\Delta+t$. Hence, $P^t_{\Delta,\Delta'}(1)$ is independent of $\Delta$ for any feasible $t$.
\end{itemize}
Combining together, we can conclude that property 1 holds.
\item \textbf{property 2}: We notice that, when $\Delta'\geq t_{max}$,
\[
P_{\Delta,\Delta'}(1) = p_{t'}P^{t'}_{\Delta,\Delta'}(1)+ p_{t^+}P^{t^+}_{\Delta,\Delta'}(1).
\]
Then, we divide the discussion into the following two cases.
\begin{itemize}
\item Since $t'=\Delta'-\Delta$, $P_{\Delta,\Delta'}^{t'}(1) = P^{t'}_{\Delta,\Delta+t'}(1)$. Then, we know that $P^{t'}_{\Delta,\Delta'}(1)$ is independent of $\Delta>0$ when $t'>0$ and $P^{t'}_{\Delta,\Delta'}(1)=0$ when $t'\leq0$ by definition. Hence, $P_{\Delta,\Delta'}^{t'}(1)$ depends on $t'$.
\item When $\Delta'\geq t_{max}$ and $\Delta'\neq\Delta+t_{max}$, $P^{t^+}_{\Delta,\Delta'}(1)=0$ for $\Delta>0$. Also, $P^{t^+}_{\Delta,\Delta'}(1)$ is independent of $\Delta>0$ when $\Delta'=\Delta+t_{max}$. Hence, $P^{t^+}_{\Delta,\Delta'}(1)$ depends only on $t'$.
\end{itemize}
Combining together, property 2 holds.
\item \textbf{property 3}: When $\Delta'>\Delta+t_{max}$, the property holds apparently. When $t_{max}-1<\Delta'<\Delta+1$, 
\[
P_{\Delta,\Delta'}(1) = p_{t'}P^{t'}_{\Delta,\Delta'}(1)+ p_{t^+}P^{t^+}_{\Delta,\Delta'}(1),
\]
where $t'\leq0$. Then, by definition, $P^{t'}_{\Delta,\Delta'}(1)=0$. Moreover, we recall that $t_{max}>1$, which indicates that $P^{t^+}_{\Delta,\Delta'}(1)=0$. Hence, property 3 holds.
\end{itemize}

\section{Proof of Theorem~\ref{prop-StationaryDistribution}}\label{pf-StationaryDistribution}
We recall that $\pi_{\Delta}$ satisfies \eqref{eq-CompactBalanceEq2} and \eqref{eq-TotalProbs}. Then, plugging in the probabilities yields the following system of linear equations.
\begin{equation}\label{eq-0util15}
\begin{split}
\pi_0 = & (1-p)\pi_0 + p\sum_{i=1}^{\tau-1}\pi_i+ \sum_{i=\tau}^{\infty}P_{i,0}(1)\pi_i \\
= & (1-p)\pi_0 + p\sum_{i=1}^{\tau-1}\pi_i+ P_{1,0}(1)\sum_{i=\tau}^{\infty}\pi_i.
\end{split}
\end{equation}
\begin{equation}\label{eq-0util16}
\pi_1 = p\pi_0 + \sum_{i=\tau}^{\infty}P_{i,1}(1)\pi_i = p\pi_0 + P_{1,1}(1)\sum_{i=\tau}^{\infty}\pi_i.
\end{equation}
For each $2\leq\Delta\leq t_{max}-1$,
\begin{equation}\label{eq-0util17}
\pi_\Delta = \begin{dcases}
(1-p)\pi_{\Delta-1} + P_{\tau,\Delta}(1)\sum_{i=\tau}^{\infty}\pi_i & \Delta-1<\tau,\\
\sum_{i=\tau}^{\Delta-1}P_{i,\Delta}(1)\pi_i + P_{\Delta,\Delta}(1)\sum_{i=\Delta}^{\infty}\pi_i & \Delta-1\geq\tau.
\end{dcases}
\end{equation}
For each $t_{max}\leq\Delta\leq \omega-1$,
\[
\pi_{\Delta} = \begin{dcases}
(1-p)\pi_{\Delta-1} & \Delta-1<\tau,\\
\sum_{i=\tau}^{\Delta-1}P_{i,\Delta}(1)\pi_i & \Delta-1\geq\tau.
\end{dcases}
\]
For each $\Delta\geq\omega$,
\begin{equation}\label{eq-0util6}
\pi_{\Delta}  = \sum_{i=\Delta-t_{max}}^{\Delta-1}P_{i,\Delta}(1)\pi_i.
\end{equation}
\[
\sum_{i=0}^{\tau-1}\pi_i + ET\sum_{i=\tau}^{\infty}\pi_i = 1.
\]
Note that we can pull the state transition probabilities in \eqref{eq-0util15}, \eqref{eq-0util16}, and \eqref{eq-0util17} out of the summation due to property 1 in Lemma~\ref{lem-MSTPAss1} and Lemma~\ref{lem-MSTPAss2}. Then, we sum \eqref{eq-0util6} over $\Delta$ from $\omega$ to $\infty$.
\begin{equation}\label{eq-0util7}
\sum_{i=\omega}^{\infty}\pi_i = \sum_{i=\omega}^{\infty}\sum_{k=i-t_{max}}^{i-1}P_{k,i}(1)\pi_k.
\end{equation}
We delve deep into the right hand side (RHS) of \eqref{eq-0util7}. To this end, we expand the first summation, which yields
\begin{align*}
RHS = & \sum_{k=\tau+1}^{\omega-1}P_{k,\omega}(1)\pi_k + \sum_{k=\tau+2}^{\omega}P_{k,\omega+1}(1)\pi_k + \cdots + \\
& \sum_{k=\omega-1}^{\omega+t_{max}-2}P_{k,\omega+t_{max}-1}(1)\pi_k +\\
& \sum_{k=\omega}^{\omega+t_{max}-1}P_{k,\omega+t_{max}}(1)\pi_k + \cdots
\end{align*}
Then, we rearrange the summation.
\begin{align*}
RHS = & P_{\tau+1,\omega}(1)\pi_{\tau+1} + \sum_{k=1}^{2}P_{\tau+2,\omega+k-1}(1)\pi_{\tau+2} + \cdots +\\
& \sum_{k=1}^{t_{max}}P_{\omega-1,\omega+k-1}(1)\pi_{\omega-1} +\\
&  \sum_{k=1}^{t_{max}}P_{\omega,\omega+k}(1)\pi_{\omega} +\sum_{k=1}^{t_{max}}P_{\omega+1,\omega+k+1}(1)\pi_{\omega+1} +\cdots
\end{align*}
Leveraging property 2 in Lemma~\ref{lem-MSTPAss1} and Lemma~\ref{lem-MSTPAss2}, we have
\begin{multline*}
RHS = \sum_{i=\tau+1}^{\omega-1}\left(\sum_{k=\tau+1}^iP_{i,t_{max}+k}(1)\right)\pi_i +\\
 \sum_{i=1}^{t_{max}}\bigg(P_{\omega,\omega+i}(1)\bigg)\left(\sum_{k=\omega}^{\infty}\pi_k\right).
\end{multline*}
We define $\Pi\triangleq \sum_{i=\omega}^{\infty}\pi_i$. Then, equation \eqref{eq-0util7} becomes the following.
\begin{equation}\label{eq-0util71}
\Pi = \sum_{i=\tau+1}^{\omega-1}\left(\sum_{k=\tau+1}^iP_{i,t_{max}+k}(1)\right)\pi_i + \sum_{i=1}^{t_{max}}\bigg(P_{\omega,\omega+i}(1)\bigg)\Pi.
\end{equation}
Finally, replacing \eqref{eq-0util6} with \eqref{eq-0util71} and applying the definition of $\Pi$ yield a system of linear equations with finite size as presented in the theorem.

\section{Proof of Corollary~\ref{cor-StationaryDistributionSpecial}}\label{pf-AoIISpecialCase1}
We start with $\tau=0$. In this case, $\omega = t_{max}+1$ and the system of linear equations becomes to the following.
\begin{multline}\label{eq-0util1}
\pi_{\Delta} = \sum_{i=0}^{\infty}P_{i,\Delta}(1)\pi_i =\\ 
\begin{dcases}
P_{0,0}(1)\pi_0 + P_{1,0}(1)\sum_{i=1}^{\infty}\pi_i & \Delta=0,\\
\sum_{i=0}^{\Delta-1}P_{i,\Delta}(1)\pi_i + P_{\Delta,\Delta}(1)\sum_{i=\Delta}^{\infty}\pi_i & 1\leq\Delta\leq t_{max}.
\end{dcases}
\end{multline}
\begin{equation}\label{eq-0util72}
\begin{split}
\Pi = & \sum_{i=1}^{t_{max}}\left(\sum_{k=1}^iP_{i,t_{max}+k}(1)\right)\pi_i + \\
& \sum_{i=1}^{t_{max}}P_{t_{max}+1,t_{max}+1+i}(1)\Pi.
\end{split}
\end{equation}
\begin{equation}\label{eq-0util2}
ET\sum_{i=0}^{\infty}\pi_i = 1.
\end{equation}
We first combine \eqref{eq-0util1} and \eqref{eq-0util2}, which yields \eqref{eq-EquivalentEq9}.
\begin{figure*}[!t]
\normalsize
\begin{equation}\label{eq-EquivalentEq9}
{\pi_{\Delta}} =
\begin{dcases}
P_{0,0}(1)\pi_0 + P_{1,0}(1)\left(\frac{1}{ET} -\pi_0\right) & \Delta=0,\\
\sum_{i=0}^{\Delta-1}P_{i,\Delta}(1)\pi_i + P_{\Delta,\Delta}(1)\left(\frac{1}{ET} - \sum_{i=0}^{\Delta-1}\pi_i\right) & 1\leq\Delta\leq t_{max}.
\end{dcases}
\end{equation}
\hrulefill
\vspace*{4pt}
\end{figure*}
Then, we have
\[
\pi_{0} = \frac{P_{1,0}(1)}{ET[1-P_{0,0}(1)+P_{1,0}(1)]}.
\]
According to \eqref{eq-0util72}, we obtain
\[
\Pi = \ddfrac{\sum_{i=1}^{t_{max}}\left(\sum_{k=1}^{i}P_{i,t_{max}+k}(1)\right)\pi_i}{1-\sum_{i=1}^{t_{max}}P_{t_{max}+1,t_{max}+1+i}(1)}.
\]

Then, we consider the case of $\tau=1$. In this case, $\omega = t_{max}+2$ and the system of linear equations reduces to the following.
\begin{equation}\label{eq-0util8}
\pi_0 = (1-p)\pi_0 + P_{1,0}(1)\sum_{i=1}^{\infty}\pi_i.
\end{equation}
\[
\pi_1 = p\pi_0 + P_{1,1}(1)\sum_{i=1}^{\infty}\pi_i.
\]
\[
\pi_\Delta = \sum_{i=1}^{\Delta-1}P_{i,\Delta}(1)\pi_i + P_{\Delta,\Delta}(1)\sum_{i=\Delta}^{\infty}\pi_i\quad 2\leq \Delta\leq t_{max}-1.
\]
\begin{equation}\label{eq-corollary1util1}
\pi_{\Delta} = \sum_{i=1}^{\Delta-1}P_{i,\Delta}(1)\pi_i\quad t_{max}\leq\Delta\leq t_{max}+1.
\end{equation}
\begin{equation}\label{eq-0util73}
\begin{split}
\Pi = & \sum_{i=2}^{t_{max}+1}\left(\sum_{k=2}^iP_{i,t_{max}+k}(1)\right)\pi_i +\\
& \sum_{i=1}^{t_{max}}P_{t_{max}+2,t_{max}+2+i}(1)\Pi.
\end{split}
\end{equation}
\begin{equation}\label{eq-0util9}
\pi_0 + ET\sum_{i=1}^{\infty}\pi_i = 1.
\end{equation}
We first combine \eqref{eq-0util8} and \eqref{eq-0util9}, which yields
\[
\pi_0 = (1-p)\pi_0 + P_{1,0}(1)\left(\frac{1-\pi_0}{ET}\right).
\]
Hence, we have
\[
\pi_0 = \frac{P_{1,0}(1)}{pET+P_{1,0}(1)}.
\]
Similarly,
\[
\pi_1 = \frac{pP_{1,0}(1)+pP_{1,1}(1)}{pET+P_{1,0}(1)}.
\]
For each $2\leq \Delta\leq t_{max}-1$,
\begin{equation}\label{eq-corollary1util0}
\pi_\Delta = \sum_{i=1}^{\Delta-1}P_{i,\Delta}(1)\pi_i + P_{\Delta,\Delta}(1)\left(\frac{1-\pi_0}{ET} - \sum_{i=1}^{\Delta-1}\pi_i\right).
\end{equation}
According to the property 3 in Lemma~\ref{lem-MSTPAss1} and Lemma~\ref{lem-MSTPAss2}, we know that $P_{\Delta,\Delta}(1) = 0$ when $t_{max}\leq\Delta\leq t_{max}+1$. Hence, we can combine \eqref{eq-corollary1util1} and \eqref{eq-corollary1util0}, for each $2\leq \Delta\leq t_{max}+1$, which yields
\[
\pi_\Delta = \sum_{i=1}^{\Delta-1}P_{i,\Delta}(1)\pi_i + P_{\Delta,\Delta}(1)\left(\frac{1-\pi_0}{ET} - \sum_{i=1}^{\Delta-1}\pi_i\right).
\]
Finally, according to \eqref{eq-0util73}, we obtain
\[
\Pi = \ddfrac{\sum_{i=2}^{t_{max}+1}\left(\sum_{k=2}^iP_{i,t_{max}+k}(1)\right)\pi_i}{1-\sum_{i=1}^{t_{max}}P_{t_{max}+2,t_{max}+2+i}(1)}.
\]

\section{Proof of Lemma~\ref{lem-CompactCost}}\label{pf-ExpectedCost}
We recall that $C^k(\Delta)$ is defined as the expected AoII $k$ time slots after the transmission starts at state $(\Delta,0,-1)$, given that the transmission is still in progress. With this in mind, we start with the case of $\Delta=0$. As AoII either increases by one or decreases to zero, we know $C^{k}(0)\in\{0,...,k\}$. Then, we distinguish between the following cases.
\begin{itemize}
\item $C^k(0) = 0$ when the receiver's estimate is correct $k$ time slots after the transmission starts. Since $\Delta = 0$, we can easily conclude that $C^k(0) = 0$ happens with probability $p^{(k)}$.
\item $C^k(0) = h$, where $1\leq h\leq k$, happens when the receiver's estimate is correct at the $(k-h)$th time slot after the transmission starts, then, the source flips the state and stays in the same state for the remaining $h-1$ time slots. Hence, $C^k(0) = h$, where $1\leq h\leq k$, happens with probability $p^{(k-h)}p(1 - p)^{h-1}$.
\end{itemize}
Combining together, we obtain
\[
C^k(0) = \sum_{h=1}^{k} hp^{(k-h)}p(1-p)^{h-1}.
\]
Then, we consider the case of $\Delta>0$. In this case, the transmission starts when the receiver's estimate is incorrect and $C^{k}(\Delta)\in\{0,1,...k-1,\Delta+k\}$. Then, we distinguish between the following cases.
\begin{itemize}
\item $C^k(\Delta) = 0$ when the receiver's estimate is correct at the $k$th time slot after the transmission starts, which happens with probability $(1-p^{(k)})$.
\item $C^k(\Delta) = h$, where $h\in\{1,2,...,k-1\}$, happens when the receiver's estimate is correct at the $(k-h)$th slot after the transmission starts. Then, the source flips the state and stays in the same state for the remaining $h-1$ time slots. Hence, $C^k(\Delta) = h$, where $h\in\{1,2,...,k-1\}$, happens with probability $(1-p^{(k-h)})p(1 - p)^{h-1}$.
\item $C^k(\Delta) = \Delta+k$ when the estimate at the receiver side is always wrong for $k$ time slots after the transmission starts. Since $\Delta > 0$ and the receiver's estimate will not change, $C^k(\Delta) = \Delta+k$ happens with probability $(1 - p)^k$. 
\end{itemize}
Combining together, for $\Delta>0$, we obtain
\[
C^k(\Delta) = \sum_{h=1}^{k-1} h(1-p^{(k-h)})p(1-p)^{h-1} + (\Delta+k)(1-p)^k.
\]

\section{Proof of Theorem~\ref{prop-LazyPerformance}}\label{pf-LazyPerformance}
We recall that when $\tau=\infty$, the transmitter will never initiate any transmissions. Hence, the receiver's estimate will never change. Without loss of generality, we assume the receiver's estimate $\hat{X}_k=0$ for all $k$. The first step in calculating the expected AoII achieved by the threshold policy with $\tau=\infty$ is to calculate the stationary distribution of the induced DTMC. We know that $\pi_{\Delta}$ satisfies the following equations.
\begin{equation}\label{eq-lazypolicy1}
\pi_0 = (1-p)\pi_0 + p\sum_{i=1}^{\infty}\pi_i.
\end{equation}
\[
\pi_1 = p\pi_0.
\]
\[
\pi_{\Delta} = (1-p)\pi_{\Delta-1}\quad \Delta\geq 2.
\]
\begin{equation}\label{eq-lazypolicy2}
\sum_{i=0}^{\infty}\pi_i = 1.
\end{equation}
Combining \eqref{eq-lazypolicy1} and \eqref{eq-lazypolicy2} yields
\[
\pi_0 = (1-p)\pi_0 + p(1-\pi_0).
\]
Hence, $\pi_0=\frac{1}{2}$. Then, we can get
\[
\pi_1 = \frac{p}{2},
\]
\[
\pi_{\Delta} = (1-p)^{\Delta-1}\pi_1 = \frac{p(1-p)^{\Delta-1}}{2}\quad \Delta\geq2.
\]
Combining together, we have
\[
\pi_0=\frac{1}{2},\quad \pi_{\Delta} = \frac{p(1-p)^{\Delta-1}}{2}\quad \Delta\geq1.
\]
Since the transmitter will never make any transmission attempts, the cost for being at state $(\Delta,0,-1)$ is nothing but $\Delta$ itself. Hence, the expected AoII is
\[
\bar{\Delta}_\infty = \sum_{\Delta=1}^{\infty}\Delta\frac{p(1-p)^{\Delta-1}}{2} = \frac{1}{2p}.
\]

\section{Proof of Theorem~\ref{thm-Sigma}}\label{pf-Performance}
We recall that, for $\Delta\geq\omega$, $\pi_\Delta$ satisfies
\begin{align*}
\pi_{\Delta} = & \sum_{i=\Delta-t_{max}}^{\Delta-1}P_{i,\Delta}(1)\pi_i \\
= & \sum_{i=1}^{t_{max}}P_{i-t_{max}+\Delta-1,\Delta}(1)\pi_{i-t_{max}+\Delta-1}\quad \Delta\geq\omega.
\end{align*}
We first focus on the system under \textbf{Assumption 1}. We know from by Lemma~\ref{lem-MSTPAss1} that $P_{\Delta,\Delta'}(1) = p_{t'}P^{t'}_{\Delta,\Delta'}(1)$ where $t' = \Delta'-\Delta$ when $\Delta'\geq\omega$. Hence, for each $\Delta\geq\omega$,
\[
\pi_{\Delta} = \sum_{i=1}^{t_{max}}p_{t_{max}+1-i}P^{t_{max}+1-i}_{i-t_{max}+\Delta-1,\Delta}(1)\pi_{i-t_{max}+\Delta-1}.
\]
Renaming the variables yields
\[
\pi_{\Delta} = \sum_{t=1}^{t_{max}}p_{t}P^{t}_{\Delta-t,\Delta}(1)\pi_{\Delta-t}\quad \Delta\geq\omega.
\]
To proceed, we define, for each $1\leq t\leq t_{max}$,
\begin{equation}\label{eq-0util41}
\pi_{\Delta,t} \triangleq p_{t}P^{t}_{\Delta-t,\Delta}(1)\pi_{\Delta-t}\quad\Delta\geq\omega.
\end{equation}
Note that $\sum_{t=1}^{t_{max}}\pi_{\Delta,t} = \pi_{\Delta}$. Then, for a given $1\leq t\leq t_{max}$, we multiple both side of  \eqref{eq-0util41} by $C(\Delta-t,1)$ and sum over $\Delta$ from $\omega$ to $\infty$. Hence, we have
\begin{equation}\label{eq-0util51}
\sum_{i=\omega}^{\infty}C(i-t,1)\pi_{i,t} = \sum_{i=\omega}^{\infty}C(i-t,1)p_tP^t_{i-t,i}(1)\pi_{i-t}.
\end{equation}
We define $\Delta_t' \triangleq C(\Delta,1) - C(\Delta-t,1)$ where $\Delta>t$. Then, according to \eqref{eq-CompactCostAss1}, we have
\[
\Delta_t' = \sum_{i=1}^{t_{max}}p_i\bigg(C^i(\Delta,1) - C^i(\Delta-t,1)\bigg).
\]
According to Lemma~\ref{lem-CompactCost}, we have
\begin{multline*}
C^{i}(\Delta-t,1) = \Delta - t + \sum_{h=1}^{i-1}\bigg(\sum_{k=1}^{h-1} k(1-p^{(h-k)})p(1-p)^{k-1} + \\
 (\Delta-t+h)(1-p)^h\bigg).
\end{multline*}
\begin{multline*}
C^{i}(\Delta,1) = \Delta+ \sum_{h=1}^{i-1}\bigg(\sum_{k=1}^{h-1} k(1-p^{(h-k)})p(1-p)^{k-1}+\\
(\Delta+h)(1-p)^h\bigg).
\end{multline*}
Subtracting the two equations yields
\begin{align*}
C^{i}(\Delta,1) - C^{i}(\Delta-t,1) = & t + \sum_{h=1}^{i-1}\bigg(t(1-p)^h\bigg) \\
= & \frac{t-t(1-p)^i}{p}.
\end{align*}
Then, we have
\[
\Delta_t' = \sum_{i=1}^{t_{max}}p_i\left(\frac{t-t(1-p)^i}{p}\right).
\]
We notice that $\Delta_t'$ is independent of $\Delta$ when $\Delta>t$. Hence, \eqref{eq-0util51} can be rewritten as
\[
\sum_{i=\omega}^{\infty}\bigg(C(i,1) - \Delta_t'\bigg)\pi_{i,t} = \sum_{i=\omega-t}^{\infty}C(i,1)p_tP^t_{i,i+t}(1)\pi_i.
\]
Then, we define $\Pi_t \triangleq \sum_{i=\omega}^{\infty}\pi_{i,t}$ and $\Sigma_t \triangleq \sum_{i=\omega}^{\infty}C(i,1)\pi_{i,t}$. We notice that $P^t_{\Delta,\Delta+t}(1)$ is independent of $\Delta$ when $\Delta>0$. Hence, we obtain 
\[
\sum_{i=\omega}^{\infty}C(i,1)\pi_{i,t} - \Delta_t'\sum_{i=\omega}^{\infty}\pi_{i,t} = p_tP^t_{1,1+t}(1)\sum_{i=\omega-t}^{\infty}C(i,1)\pi_i.
\]
Plugging in the definitions yields
\[
\Sigma_t - \Delta_t'\Pi_t = p_tP^t_{1,1+t}(1)\left(\sum_{i=\omega-t}^{\omega-1}C(i,1)\pi_i+\Sigma\right).
\]
Summing the above equation over $t$ from $1$ to $t_{max}$ yields
\begin{multline*}
\sum_{t=1}^{t_{max}}\bigg(\Sigma_t - \Delta_t'\Pi_t\bigg) = \\
\sum_{t=1}^{t_{max}}\left[p_tP^t_{1,1+t}(1)\left(\sum_{i=\omega-t}^{\omega-1}C(i,1)\pi_i+\Sigma\right)\right].
\end{multline*}
Rearranging the above equation yields
\begin{multline*}
\Sigma - \sum_{t=1}^{t_{max}}\Delta_t'\Pi_t = \sum_{t=1}^{t_{max}}\left[p_tP^t_{1,1+t}(1)\left(\sum_{i=\omega-t}^{\omega-1}C(i,1)\pi_i\right)\right] +\\
\sum_{t=1}^{t_{max}}\bigg(p_tP^t_{1,1+t}(1)\bigg)\Sigma.
\end{multline*}
Hence, the closed-form expression of $\Sigma$ is
\[
\Sigma = \ddfrac{\sum_{t=1}^{t_{max}}\left[p_tP^t_{1,1+t}(1)\left(\sum_{i=\omega-t}^{\omega-1}C(i,1)\pi_i\right) + \Delta_t'\Pi_t\right]}{1-\sum_{t=1}^{t_{max}}\bigg(p_tP^t_{1,1+t}(1)\bigg)}.
\]
In the following, we calculate $\Pi_t$. Combining the definition of $\Pi_t$ with \eqref{eq-0util41}, we have
\begin{align*}
\Pi_t\triangleq \sum_{i=\omega}^{\infty}\pi_{i,t} = & \sum_{i=\omega}^{\infty}\bigg(p_{t}P^{t}_{i-t,i}(1)\pi_{i-t}\bigg)\\
= & \sum_{i=\omega-t}^{\infty}\bigg(p_{t}P^{t}_{i,i+t}(1)\pi_{i}\bigg).
\end{align*}
Since $P^{t}_{\Delta,\Delta+t}(1)$ is independent of $\Delta$ when $\Delta>0$, we have
\[
\Pi_t = p_{t}P^{t}_{1,1+t}(1)\left(\sum_{i=\omega-t}^{\omega-1}\pi_{i} + \Pi\right).
\]
Combining together, we recover the results for \textbf{Assumptio 1} as presented in the first part of the theorem.

In the sequel, we focus on \textbf{Assumption 2}. To this end, we follow similar steps as detailed above. We recall from Lemma~\ref{lem-MSTPAss2}, $P_{\Delta,\Delta'}(1) = p_{t'}P^{t'}_{\Delta,\Delta'}(1) + p_{t^+}P^{t^+}_{\Delta,\Delta'}(1)$ where $t' = \Delta'-\Delta$ when $\Delta'\geq\omega$. Then, for each $\Delta\geq\omega$,
\begin{multline*}
\pi_\Delta = \sum_{i=1}^{t_{max}}\bigg(p_{t_{max}+1-i}P^{t_{max}+1-i}_{\Delta-t_{max}+i-1,\Delta}(1) + \\
p_{t^+}P^{t^+}_{\Delta-t_{max}+i-1,\Delta}(1)\bigg)\pi_{\Delta-t_{max}-1+i}.
\end{multline*}
Renaming the variables yields
\begin{align*}
\pi_\Delta = & \sum_{t=1}^{t_{max}}\bigg(p_tP^t_{\Delta-t,\Delta}(1) + p_{t^+}P^{t^+}_{\Delta-t,\Delta}(1)\bigg)\pi_{\Delta-t}\\
= & \sum_{t=1}^{t_{max}}\Upsilon(\Delta,t)\pi_{\Delta-t}\quad \Delta\geq\omega,
\end{align*}
where $\Upsilon(\Delta,t)\triangleq p_tP^t_{\Delta-t,\Delta}(1) + p_{t^+}P^{t^+}_{\Delta-t,\Delta}(1)$. We notice that $\Upsilon(\Delta,t)$ is independent of $\Delta$ when $\Delta \geq \omega$. To proceed, we define, for each $1\leq t\leq t_{max}$,
\[
\pi_{\Delta,t} \triangleq \Upsilon(\Delta,t)\pi_{\Delta-t}\quad\Delta\geq\omega.
\]
Note that $\sum_{t=1}^{t_{max}}\pi_{\Delta,t} = \pi_{\Delta}$. Then, for a given $1\leq t\leq t_{max}$, we have
\begin{equation}\label{eq-ZeroWaitPolicy4}
\sum_{i=\omega}^{\infty}C(i-t,1)\pi_{i,t} = \sum_{i=\omega}^{\infty}C(i-t,1)\Upsilon(i,t)\pi_{i-t}.
\end{equation}
We define $\Delta_t' \triangleq C(\Delta,1) - C(\Delta-t,1)$ where $\Delta>t$. Then, according to \eqref{eq-CompactCostAss2}, we have
\begin{multline*}
\Delta_t' = \sum_{i=1}^{t_{max}}p_i\bigg(C^i(\Delta,1) - C^i(\Delta-t,1)\bigg) + \\
p_{t^+}\bigg(C^{t_{max}}(\Delta,1) - C^{t_{max}}(\Delta-t,1)\bigg).
\end{multline*}
By Lemma~\ref{lem-CompactCost}, we have
\begin{multline*}
C^{i}(\Delta-t,1) = \Delta - t + \sum_{h=1}^{i-1}\bigg(\sum_{k=1}^{h-1} k(1-p^{(h-k)})p(1-p)^{k-1}+\\
(\Delta-t+h)(1-p)^h\bigg).
\end{multline*}
\begin{multline*}
C^{i}(\Delta,1) = \Delta+ \sum_{h=1}^{i-1}\bigg(\sum_{k=1}^{h-1} k(1-p^{(h-k)})p(1-p)^{k-1}+\\
(\Delta+h)(1-p)^h\bigg).
\end{multline*}
Subtracting the two equations yields
\begin{align*}
C^{i}(\Delta,1) - C^{i}(\Delta-t,1) & = t + \sum_{h=1}^{i-1}\bigg(t(1-p)^h\bigg)\\
& = \frac{t-t(1-p)^i}{p}.
\end{align*}
Then, for each $1\leq t\leq t_{max}$, we have
\[
\Delta_t' = \sum_{i=1}^{t_{max}}p_i\left(\frac{t-t(1-p)^i}{p}\right) + p_{t^+}\left(\frac{t-t(1-p)^{t_{max}}}{p}\right).
\]
We notice that $\Delta_t' = C(\Delta,1) - C(\Delta-t,1)$ is independent of $\Delta$ when $\Delta>t$. Hence, equation \eqref{eq-ZeroWaitPolicy4} can be written as
\[
\sum_{i=\omega}^{\infty}\bigg(C(i,1) - \Delta_t'\bigg)\pi_{i,t} = \sum_{i=\omega-t}^{\infty}C(i,1)\Upsilon(i+t,t)\pi_i.
\]
Then, we define $\Pi_t \triangleq \sum_{i=\omega}^{\infty}\pi_{i,t}$ and $\Sigma_t \triangleq \sum_{i=\omega}^{\infty}C(i,1)\pi_{i,t}$. We recall that $\Upsilon(\Delta,t)$ is independent of $\Delta$ when $\Delta\geq \omega$. Hence, plugging in the definitions yields
\[
\Sigma_t - \Delta_t'\Pi_t = \sum_{i=\omega-t}^{\omega-1}\Upsilon(i+t,t)C(i,1)\pi_i+\Upsilon(\omega+t,t)\Sigma.
\]
Summing the above equation over $t$ from $1$ to $t_{max}$ yields
\begin{multline*}
\sum_{t=1}^{t_{max}}\bigg(\Sigma_t - \Delta_t'\Pi_t\bigg) =\\
\sum_{t=1}^{t_{max}}\left(\sum_{i=\omega-t}^{\omega-1}\Upsilon(i+t,t)C(i,1)\pi_i+\Upsilon(\omega+t,t)\Sigma\right).
\end{multline*}
Rearranging the above equation yields
\begin{align*}
\Sigma - \sum_{t=1}^{t_{max}}\Delta_t'\Pi_t = & \sum_{t=1}^{t_{max}}\left(\sum_{i=\omega-t}^{\omega-1}\Upsilon(i+t,t)C(i,1)\pi_i\right) +\\
& \sum_{t=1}^{t_{max}}\Upsilon(\omega+t,t)\Sigma.
\end{align*}
Then, the closed-form expression of $\Sigma$ is
\[
\Sigma = \ddfrac{\sum_{t=1}^{t_{max}}\left[\left(\sum_{i=\omega-t}^{\omega-1}\Upsilon(i+t,t)C(i,1)\pi_i\right) + \Delta_t'\Pi_t\right]}{1-\sum_{t=1}^{t_{max}}\Upsilon(\omega+t,t)}.
\]
In the following, we calculate $\Pi_t$. We have
\[
\Pi_t\triangleq \sum_{i=\omega}^{\infty}\pi_{i,t} = \sum_{i=\omega}^{\infty}\Upsilon(i,t)\pi_{i-t} = \sum_{i=\omega-t}^{\infty}\Upsilon(i+t,t)\pi_i.
\]
Since $\Upsilon(\Delta,t)$ is independent of $\Delta$ if $\Delta \geq \omega$, we have
\[
\Pi_t = \sum_{i=\omega-t}^{\omega-1}\Upsilon(i+t,t)\pi_i + \Upsilon(\omega+t,t)\Pi\quad 1\leq t\leq t_{max}.
\]
Combining together, we recover the results for the system under \textbf{Assumption 2} as presented in the second half of the theorem.

\section{Proof of Lemma~\ref{lem-Monotone}}\label{pf-Monotone}
Leveraging Lemma~\ref{lem-Converge}, the result can be proved using mathematical induction. To start with, we initialize $V_{\gamma,0}(s) = 0$ for all $s$. Hence, the base case (i.e., $\nu=0$) is true. Then, we assume the monotonicity holds at iteration $\nu$, and check whether the monotonicity still holds at iteration $\nu+1$. We recall that the estimated value function $V_{\gamma,\nu+1}(s)$ is updated using \eqref{eq-ValueIterationGammaDiscount}. Hence, the structural property is embedded in the state transition probability $P_{s,s'}(a)$. Using the state transition probabilities in Appendix~\ref{app-STP}, equation \eqref{eq-ValueIterationGammaDiscount} for the state with $\Delta>0$ can be written as \eqref{eq-EquivalentEq5}.
\begin{figure*}[!t]
\normalsize
\begin{equation}\label{eq-EquivalentEq5}
\begin{split}
V_{\gamma,\nu+1}(\Delta,t,i) = & \min_{a\in\{0,1\}}\left\lbrace \Delta+\gamma\sum_{\Delta',t',i'}Pr[(\Delta',t',i')\mid (\Delta,t,i),a]V_{\gamma,\nu}(\Delta',t',i')\right\rbrace\\
= & \min_{a\in\{0,1\}}\bigg\lbrace \Delta+\gamma\sum_{t',i'}\bigg[Pr[(\Delta+1,t',i')\mid (\Delta,t,i),a]V_{\gamma,\nu}(\Delta+1,t',i') +\\
& \hspace{9em} Pr[(0,t',i')\mid (\Delta,t,i),a]V_{\gamma,\nu}(0,t',i')\bigg]\bigg\rbrace.
\end{split}
\end{equation}
\hrulefill
\vspace*{4pt}
\end{figure*}
Moreover, for any $\Delta_1>0$ and $\Delta_2>0$,
\begin{multline*}
Pr[(\Delta_1+1,t',i')\mid (\Delta_1,t,i),a] = \\
Pr[(\Delta_2+1,t',i')\mid (\Delta_2,t,i),a].
\end{multline*}
\[
Pr[(0,t',i')\mid (\Delta_1,t,i),a] = Pr[(0,t',i')\mid (\Delta_2,t,i),a].
\]
Let $V^a_{\gamma,\nu+1}(\Delta,t,i)$ be the resulting $V_{\gamma,\nu+1}(\Delta,t,i)$ when action $a$ is chosen. Then, we have \eqref{eq-EquivalentEq6} holds. 
\begin{figure*}[!t]
\normalsize
\begin{multline}\label{eq-EquivalentEq6}
V^a_{\gamma,\nu+1}(\Delta+1,t,i) - V^a_{\gamma,\nu+1}(\Delta,t,i) = \\
1 +\gamma\sum_{t',i'}\bigg\lbrace Pr[(\Delta+1,t',i')\mid (\Delta,t,i),a]\bigg(V_{\gamma,\nu}(\Delta+2,t',i') -V_{\gamma,\nu}(\Delta+1,t',i')\bigg)\bigg\rbrace.
\end{multline}
\hrulefill
\vspace*{4pt}
\end{figure*}
Combining with the assumption for iteration $\nu$, we can easily conclude that $V^a_{\gamma,\nu+1}(\Delta+1,t,i) \geq V^a_{\gamma,\nu+1}(\Delta,t,i)$ when $\Delta>0$ for $a\in\{0,1\}$. Since, $V_{\gamma,\nu+1}(\Delta,t,i) =\min_{a\in\{0,1\}}\{V^a_{\gamma,\nu+1}(\Delta,t,i)\}$, we know that $V_{\gamma,\nu+1}(\Delta+1,t,i) \geq V_{\gamma,\nu+1}(\Delta,t,i)$ when $\Delta>0$. Finally, by mathematical induction, we can conclude that Lemma~\ref{lem-Monotone} is true.

\section{Proof of Theorem~\ref{thm-optimalexist}}\label{pf-optimalexist}
We first define $h_{\gamma}(s)\triangleq V_{\gamma}(s)- V_{\gamma}(s^{ref})$ as the relative value function and choose the reference state $s^{ref}=(0,0,-1)$. For simplicity, we abbreviate the reference state $s^{ref}$ as $0$ for the remainder of this proof. Then, we show that $\mathcal{M}$ verifies the two conditions given in~\cite{b17}. As a result, the existence of the optimal policy is guaranteed.
\begin{enumerate}
\item \textit{There exists a non-negative $N$ such that $-N\leq h_{\gamma}(s)$ for all $s$ and $\gamma$}: Leveraging Lemma~\ref{lem-Monotone}, we can easily conclude that $h_{\gamma}(s)$ is also non-decreasing in $\Delta$ when $\Delta>0$. In the following, we consider the policy $\phi$ being the threshold policy with $\tau=0$. Then, we know that policy $\phi$ induces an irreducible ergodic Markov chain and the expected cost is finite. Let $c_{s,s'}(\phi)$ be the expected cost of a first passage from $s\in\mathcal{S}$ to $s'\in\mathcal{S}$ when policy $\phi$ is adopted.  Then, by~\cite[Proposition 4]{b17}, we know that $c_{s,0}(\phi)$ is finite. Meanwhile, $h_{\gamma}(s)\leq c_{s,0}(\phi)$ as is given in the proof of~\cite[Proposition 5]{b17}. Hence, we have $V_{\gamma}(0) - V_{\gamma}(s) \leq c_{0,s}(\phi)$ and $V_{\gamma}(0) - V_{\gamma}(s) = -h_{\gamma}(s)$. Hence, we have $h_{\gamma}(s) \geq -c_{0,s}(\phi)$. Combining with the monotonicity proved in Lemma~\ref{lem-Monotone}, we can choose $-N=\min_{s\in G}\{c_{0,s}(\phi)\}$, where $G=\{s:\Delta\in\{0,1\}\}$. This condition indicates that~\cite[Assumption 2]{b17} holds.
\item \textit{$\mathcal{M}$ has a stationary policy $\phi$ inducing an irreducible, ergodic Markov chain. Moreover, the resulting expected cost is finite}: We consider the policy $\phi$ being the threshold policy with $\tau=0$. Then, according to Section~\ref{sec-PolicyPerformance}, it induces an irreducible, ergodic Markov chain and the resulting expected cost is finite. Then, according to~\cite[Proposition 5]{b17}, we can conclude that~\cite[Assumptions 1 and 3]{b17} hold.
\end{enumerate}
As the two conditions are verified, the existence of the optimal policy is guaranteed by~\cite[Theorem]{b17}. Moreover, the minimum expected cost is independent of the initial state.

\section{Proof of Theorem~\ref{thm-ASMConvergence}}\label{pf-ASMConvergence}
We inherit the definitions and notations introduced in Section~\ref{sec-OptimalPolicyExistence}. We further define $v_{\gamma,n}(\cdot)$ as the minimum expected $\gamma$-discounted cost for operating the system from time $0$ to time $n-1$. It is known that $\lim_{n\rightarrow\infty}v_{\gamma,n}(s) = V_{\gamma}(s)$, for all $s\in\mathcal{S}$. We also define the expected cost under policy $\phi$ as
\[
J_{\phi}(s) = \limsup_{K\rightarrow\infty}\frac{1}{K}\mathbb{E}_{\phi}\left(\sum_{k=0}^{K-1}C(s_k)\mid s_0\right),
\]
and $J(s) \triangleq \inf_{\phi} J_{\phi}(s)$ is the best that can be achieved. $V_{\phi,\gamma}^{(m)}(s)$, $V_{\gamma}^{(m)}(s)$, $v_{\gamma,n}^{(m)}(s)$, $J^{(m)}_{\phi}(s)$, $J^{(m)}(s)$, and $h_{\gamma}^{(m)}(s)$ are defined analogously for $\mathcal{M}^{(m)}$. With the above definitions in mind, we show that our system verifies the two assumptions given in~\cite{b23}.
\begin{itemize}
\item \textit{Assumption 1}: There exists a non-negative (finite) constant $L$, a non-negative (finite) function $M(\cdot)$ on $\mathcal{S}$, and constants $m_0$ and $\gamma_0\in [0, 1)$, such that $-L \leq h_{\gamma}^{(m)}(s) \leq M(s)$, for $s\in \mathcal{S}^{(m)}$, $m \geq m_0$, and $\gamma\in (\gamma_0, 1)$: $L$ can be chosen in the same way as presented in the proof of Theorem~\ref{thm-ASMConvergence}. More precisely, $-L =\min_{s\in G}\{h_{\gamma}^{(m)}(s)\}$, where $G=\{s:\Delta\in\{0,1\}\}$. Let $c_{s,0}(\phi)$ be the expected cost of a first passage from $s\in\mathcal{S}$ to the reference state $0$ when policy $\phi$ is adopted and $c_{x,0}^{(m)}(\phi)$ is defined analogously for $\mathcal{M}^{(m)}$. In the following, we consider the policy $\phi$ being the threshold policy with $\tau=\infty$. We recall from Section~\ref{sec-PolicyPerformance}  that the policy $\phi$ induces an irreducible ergodic Markov chain, and the expected cost is finite. Hence, $h_{\gamma}^{(m)}(s)\leq c_{s,0}^{(m)}(\phi)$ by~\cite[Proposition 5]{b17} and $c_{x,0}(\phi)$ is finite by~\cite[Proposition 4]{b17}. We also know from the proof of ~\cite[Corollary 4.3]{b23} that $c_{s,0}(\phi)$ satisfies the following equation.
\begin{equation}\label{eq-c}
c_{s,0}(\phi) = C(s) + \sum_{s'\in\mathcal{S}-\{0\}}P^{\phi}_{ss'}c_{s',0}(\phi),
\end{equation}
where $P^{\phi}_{ss'}$ is the state transition probability from state $s$ to $s'$ under policy $\phi$ for $\mathcal{M}$. $P^{(m),\phi}_{ss'}$ is defined analogously for $\mathcal{M}^{(m)}$. We can verify in a similar way to the proof of Lemma~\ref{lem-Monotone} that $c_{s,0}(\phi)$ is non-decreasing in $\Delta>0$. The proof is omitted here for the sake of space. Then, we have \eqref{eq-ASMInequality} holds
\begin{figure*}[!t]
\normalsize
\begin{equation}\label{eq-ASMInequality}
\begin{split}
\sum_{y\in\mathcal{S}^{(m)}_{-1}}P_{sy}^{(m),\phi}c_{y,0}(\phi) & = \sum_{y\in\mathcal{S}^{(m)}_{-1}}P^{\phi}_{sy}c_{y,0}(\phi)+\sum_{y\in\mathcal{S}^{(m)}_{-1}}\left(\sum_{z\in\mathcal{S}\setminus\mathcal{S}^{(m)}}P^{\phi}_{sz}q_z(y)\right)c_{y,0}(\phi)\\
& = \sum_{y\in\mathcal{S}^{(m)}_{-1}}P^{\phi}_{sy}c_{y,0}(\phi)+\sum_{z\in\mathcal{S}\setminus\mathcal{S}^{(m)}}P^{\phi}_{sz}\left(\sum_{y\in\mathcal{S}^{(m)}_{-1}}q_z(y)c_{y,0}(\phi)\right)\\
& \leq \sum_{y\in\mathcal{S}^{(m)}_{-1}}P^{\phi}_{sy}c_{y,0}(\phi)+\sum_{z\in\mathcal{S}\setminus\mathcal{S}^{(m)}}P^{\phi}_{sz}c_{z,0}(\phi)\\
& = \sum_{y\in\mathcal{S}\setminus\{0\}}P^{\phi}_{sy}c_{y,0}(\phi).
\end{split}
\end{equation}
\hrulefill
\vspace*{4pt}
\end{figure*}
where $\mathcal{S}^{(m)}_{-1} = \mathcal{S}^{(m)}\setminus\{0\}$ and $q_{s'}(s)=\mathbbm{1}\{t'=t;i'=i\}$, which is an indicator function with value $1$ when the transitions to state $s'$ are redirected to state $s$. Otherwise, $q_{s'}(s)=0$. Moreover, $\sum_{s\in\mathcal{S}^{(m)}_{-1}}q_{s'}(s)=1$. Applying \eqref{eq-ASMInequality} to \eqref{eq-c} yields
\begin{equation*}
c_{s,0}(\phi) \geq C(s) + \sum_{y\in\mathcal{S}^{(m)}-\{0\}}P_{sy}^{(m),\phi}c_{y,0}(\phi).
\end{equation*}
Bearing in mind that $c_{s,0}^{(m)}(\phi)$ satisfies the following.
\begin{equation*}
c_{s,0}^{(m)}(\phi) = C(s) + \sum_{y\in\mathcal{S}^{(m)}-\{0\}}P_{sy}^{(m),\phi}c_{y,0}^{(m)}(\phi).
\end{equation*}
Hence, we can conclude that $c_{s,0}^{(m)}(\phi)\leq c_{s,0}(\phi)$. Then, we can choose $M(s)=c_{s,0}(\phi)<\infty$.

\item \textit{Assumption 2}: $\limsup_{m\rightarrow\infty}J^{(m)} \triangleq J^*<\infty$ and $J^*\leq J(s)$ for all $s\in\mathcal{S}$: We first show that~\cite[Proposition 5.1]{b23} is true. Since we redistribute the transitions in a way such that, for each $s'\in\mathcal{S}\setminus\mathcal{S}^{(m)}$,
\begin{equation*}
\sum_{y\in\mathcal{S}^{(m)}}q_{s'}(y)v_{\gamma,n}(y) = v_{\gamma,n}(s),
\end{equation*}
where $s=(m,t',i')$. Hence, we only need to verify that, for each $s'\in\mathcal{S}\setminus\mathcal{S}^{(m)}$ and $s=(m,t',i')$,
\begin{equation}\label{eq-inequalityv}
v_{\gamma,n}(s) \leq v_{\gamma,n}(s').
\end{equation}
To this end, we notice that $v_{\gamma,n}(s)$ satisfies the following inductive form~\cite{b23}.
\begin{equation*}
v_{\gamma,n+1}(s) = \min_{a}\left\lbrace C(s)+\gamma\sum_{s'\in\mathcal{S}}P_{s,s'}(a)v_{\gamma,n}(s')\right\rbrace.
\end{equation*}
By following similar steps to those in the proof of Lemma~\ref{lem-Monotone}, we can prove the monotonicity of $v_{\gamma,n}(s)$ for $\Delta>0$ and $n\geq0$. The proof is omitted for the sake of space. Hence, \eqref{eq-inequalityv} is true since $\Delta'>m>0$. Apparently, $J(s)$ is finite for $s\in\mathcal{S}$. Then, according to~\cite[Corollary 5.2]{b23}, assumption 2 is true.
\end{itemize}
Consequently, by~\cite[Theorem 2.2]{b23}, we know
\begin{itemize}
\item There exists an average cost optimal stationary policy for $\mathcal{M}^{(m)}$.
\item Any limit point of the sequence of optimal policies for $\mathcal{M}^{(m)}$ is optimal for $\mathcal{M}$.
\end{itemize}

\section{Proof of Theorem~\ref{thm-PIT}}\label{pf-PIT}
The proof is based on~\cite[pp.~42-43]{b18}. We consider a generic MDP $\mathcal{M} = (\mathcal{S},\mathcal{A},\mathcal{P},\mathcal{C})$. Let $C(s,A)$ be the instant cost for being at state $s\in\mathcal{S}$ under policy $A$. We also define $P^A_{s,s'}$ as the probability that applying policy $A$ at state $s$ will lead to state $s'$. Finally, $V^A(s)$ is defined as the value function resulting from the operation of policy $A$. Since $B$ is chosen over $A$, we have
\[
C(s,B) + \sum_{s'\in\mathcal{S}}P^B_{s,s'}V^A(s')\leq C(s,A) + \sum_{s'\in\mathcal{S}}P^A_{s,s'}V^A(s').
\]
Then, for each $s\in\mathcal{S}$, we define
\begin{multline*}
\gamma_s \triangleq C(s,B) + \sum_{s'\in\mathcal{S}}P^B_{s,s'}V^A(s') - \\
C(s,A) - \sum_{s'\in\mathcal{S}}P^A_{s,s'}V^A(s')\leq 0.
\end{multline*}
Meanwhile, both policies satisfy their own Bellman equation.
\[
V^A(s) + \theta^A = C(s,A) + \sum_{s'\in\mathcal{S}}P^A_{s,s'}V^A(s')\quad s\in\mathcal{S},
\]
\[
V^B(s) + \theta^B = C(s,B) + \sum_{s'\in\mathcal{S}}P^B_{s,s'}V^B(s')\quad s\in\mathcal{S},
\]
where $\theta^A$ and $\theta^B$ are the expected costs resulting from the operation of policy $A$ and policy $B$, respectively. Then, subtracting the two expressions and bringing in the expression for $\gamma_s$ yield
\[
V^B(s) - V^A(s) + \theta^B - \theta^A = \gamma_s + \sum_{s'\in\mathcal{S}}P^B_{s,s'}(V^B(s')-V^A(s')).
\]
Let $V^{\Delta}(s) \triangleq V^B(s) - V^A(s)$ and $\theta^{\Delta}\triangleq\theta^B - \theta^A$. Then, we have
\[
V^{\Delta}(s) + \theta^{\Delta} = \gamma_s + \sum_{s'\in\mathcal{S}}P^B_{s,s'}V^{\Delta}(s')\quad s\in\mathcal{S}.
\]
We know that
\[
\theta^{\Delta} = \sum_{s\in\mathcal{S}}\pi^B_s\gamma_s,
\]
where $\pi^B_s$ is the steady-state probability of state $s$ under policy $B$. Since $\pi^B_s$ is non-negative and $\gamma_s$ is non-positive, we can conclude that $\theta^{\Delta}\leq0$. Consequently, $\theta^B\leq\theta^A$.

Then, we prove that the resulting policy is optimal when the policy improvement step converges. We prove this by contradiction. We assume that there are two policies $A$ and $B$ that satisfy $\theta^B<\theta^A$. Meanwhile, the policy improvement step has converged to policy $A$. Since the policy has converged, we know that $\gamma_s\geq0$ for all $s\in\mathcal{S}$. Hence, $\theta^{\Delta}\geq0$. Then, according to the definition of $\theta^{\Delta}$, we have $\theta^B\geq\theta^A$, which contradicts the assumption. Hence, superior policies cannot remain undiscovered. Then, we can conclude that the resulting policy is optimal when the policy iteration algorithm converges.

\section{Proof of Theorem~\ref{thm-optimalpolicy}}\label{pf-optimalpolicy}
The general procedure for the optimality proof can be summarized as follows.
\begin{enumerate}
\item \textit{Policy Evaluation}: We calculate the value function resulting from the adoption of the threshold policy with $\tau=1$.
\item \textit{Policy Improvement}: We obtain a new policy using the value function obtained in the previous step and verify that the new policy is the threshold policy with $\tau=1$.
\end{enumerate}
In the following, we elaborate on these two steps.
\paragraph{Policy Evaluation} We first calculate the value function under the threshold policy with $\tau=1$. For simplicity of notation, we denote the policy as $\phi$. Let $V^\phi(\Delta)$ be the value function of state $(\Delta,0,-1)$ under the policy $\phi$. Then, combining \eqref{eq-CompactBellman} with the expression of $P_{\Delta,\Delta'}(a)$ in Lemma~\ref{lem-MSTPAss1} and Lemma~\ref{lem-MSTPAss2}, $V^\phi(\Delta)$ satisfies the following system of linear equations.
\begin{equation}\label{eq-valueiterationGen}
V^{\phi}(0) = - \theta^{\phi} + pV^{\phi}(1) + (1-p)V^{\phi}(0).
\end{equation}
For \textbf{Assumption 1} and each $\Delta\geq1$,
\begin{multline*}
V^{\phi}(\Delta) = C(\Delta,1) - ET\theta^{\phi} +\\
\sum_{t=1}^{t_{max}}\left[p_t\left( \sum_{k=0}^{t-1}P^t_{\Delta,k}(1)V^{\phi}(k) + P^t_{\Delta,\Delta+t}(1)V^{\phi}(\Delta+t)\right)\right],
\end{multline*}
and, for \textbf{Assumption 2} and each $\Delta\geq1$, we have \eqref{eq-EquivalentEq10} holds
\begin{figure*}[!t]
\normalsize
\begin{multline}\label{eq-EquivalentEq10}
V^{\phi}(\Delta) = C(\Delta,1) - ET\theta^{\phi} + \sum_{t=1}^{t_{max}}\left[p_t\left( \sum_{k=0}^{t-1}P^t_{\Delta,k}(1)V^{\phi}(k) + P^t_{\Delta,\Delta+t}(1)V^{\phi}(\Delta+t)\right)\right] +\\
p_{t^+}\left( \sum_{k=0}^{t_{max}-1}P^{t^+}_{\Delta,k}(1)V^{\phi}(k) + P^{t^+}_{\Delta,\Delta+t_{max}}(1)V^{\phi}(\Delta+t_{max})\right).
\end{multline}
\hrulefill
\vspace*{4pt}
\end{figure*}
where $\theta^{\phi}$ is the expected AoII resulting from the adoption of $\phi$. It is difficult to solve the above system of linear equations directly for the exact solution. However, as we will see later, some structural properties of the value function are sufficient. These properties are summarized in the following lemma.
\begin{lemma}\label{lem-VDeltaproperty}
$V^{\phi}(\Delta)$ satisfies the following equations.
\[
V^{\phi}(1) - V^{\phi}(0) = \frac{\theta^{\phi}}{p},
\]
\[
V^{\phi}(\Delta+1) - V^{\phi}(\Delta) = \sigma\quad \Delta\geq1,
\]
where for \textbf{Assumption 1},
\[
\sigma = \ddfrac{\sum_{t=1}^{t_{max}}p_t\left(\frac{1-(1-p)^{t}}{p}\right)}{1-\sum_{t=1}^{t_{max}}pp_t(1-p)^{t-1}},
\]
and, for \textbf{Assumption 2},
\[
\sigma = \ddfrac{\sum_{t=1}^{t_{max}}p_t\left(\frac{1-(1-p)^t}{p}\right) + p_{t^+}\left(\frac{1-(1-p)^{t_{max}}}{p}\right)}{1-\left(\sum_{t=1}^{t_{max}}pp_t(1-p)^{t-1} + p_{t^+}(1-p)^{t_{max}}\right)}.
\]
\end{lemma}
\begin{proof}
First of all, from \eqref{eq-valueiterationGen}, we can easily obtain
\[
V^{\phi}(1) - V^{\phi}(0) = \frac{\theta^{\phi}}{p}.
\]
Then, we show that $V^{\phi}(\Delta+1) - V^{\phi}(\Delta)$ is constant for $\Delta\geq1$. We start with \textbf{Assumption 1}. According to Theorem~\ref{thm-optimalexist}, the optimal policy exists. Hence, the iterative policy evaluation algorithm~\cite[pp.74]{b19} can be used to solve the system of linear equations for $V^{\phi}(\Delta)$. Let $V_{\nu}^{\phi}(\Delta)$ be the estimated value function at iteration $\nu$ of the iterative policy evaluation algorithm. Without loss of generality, we initialize $V^{\phi}_0(\Delta)=0$ for all $\Delta$. Then, for each $\Delta\geq1$, the value function is updated in the following way.
\begin{multline*}
V^{\phi}_{\nu+1}(\Delta) = C(\Delta,1) - ET\theta^{\phi} +\\
\sum_{t=1}^{t_{max}}\left[p_t\left( \sum_{k=0}^{t-1}P^t_{\Delta,k}(1)V^{\phi}_\nu(k) + P^t_{\Delta,\Delta+t}(1)V^{\phi}_\nu(\Delta+t)\right)\right].
\end{multline*}
Then, we have $\lim_{\nu\rightarrow\infty}V^\phi_\nu(\Delta) = V^\phi(\Delta)$. Hence, we can prove the desired results using mathematical induction. The base case $\nu=0$ is true by initialization. Then, we assume $V^\phi_\nu(\Delta+1) - V^\phi_\nu(\Delta)=\sigma_\nu$ where $\sigma_\nu$ is independent of $\Delta\geq1$. Then, we will exam whether $V^\phi_{\nu+1}(\Delta+1) - V^\phi_{\nu+1}(\Delta)$ is independent of $\Delta\geq1$. Leveraging the properties in Lemma~\ref{lem-MSTPAss1}, we have \eqref{eq-EquivalentEq8} holds.
\begin{figure*}[!t]
\normalsize
\begin{equation}\label{eq-EquivalentEq8}
\begin{split}
V^\phi_{\nu+1}(\Delta+1) - & V^\phi_{\nu+1}(\Delta)\\
= & C(\Delta+1,1) - ET\theta^\phi + \sum_{t=1}^{t_{max}}\left[p_t\left(\sum_{k=0}^{t-1}P^t_{\Delta+1,k}(1)V^\phi_\nu(k) + P^t_{\Delta+1,\Delta+1+t}(1)V^\phi_\nu(\Delta+t+1)\right)\right] -\\
& C(\Delta,1) + ET\theta^\phi - \sum_{t=1}^{t_{max}}\left[p_t\left(\sum_{k=0}^{t-1}P^t_{\Delta,k}(1)V^\phi_\nu(k) + P^t_{\Delta,\Delta+t}(1)V^\phi_\nu(\Delta+t)\right)\right]\\
= & C(\Delta+1,1) - C(\Delta,1) + \sum_{t=1}^{t_{max}}\bigg(p_tP^t_{\Delta,\Delta+t}(1)\sigma_\nu\bigg).
\end{split}
\end{equation}
\hrulefill
\vspace*{4pt}
\end{figure*}
According to Lemma~\ref{lem-CompactCost}, we have
\[
C(\Delta+1,1) - C(\Delta,1) = \sum_{t=1}^{t_{max}}\bigg(C^{t}(\Delta+1,1) - C^{t}(\Delta,1)\bigg).
\]
In the case of $\Delta\geq1$, we have
\begin{multline*}
C^{t}(\Delta+1,1) - C^{t}(\Delta,1) = \\
1 + \sum_{k=1}^{t-1}\bigg((k+\Delta+1)(1-p)^k - (k+\Delta)(1-p)^k\bigg)\\
= \frac{1-(1-p)^{t}}{p}\quad  1\leq t\leq t_{max}.
\end{multline*}
Combining together, we obtain
\[
C(\Delta+1,1) - C(\Delta,1) = \sum_{t=1}^{t_{max}}\left(p_t\frac{1-(1-p)^{t}}{p}\right).
\]
Hence, we can conclude that $V^\phi_{\nu+1}(\Delta+1) - V^\phi_{\nu+1}(\Delta)$ is independent of $\Delta$ when $\Delta\geq1$. Then, by mathematical induction, $V^\phi(\Delta) - V^\phi(\Delta+1)$ is independent of $\Delta$ when $\Delta\geq1$. We denote by $\sigma$ the constant. Then, $\sigma$ satisfies the following equation.
\begin{align*}
\sigma  =& V^\phi(\Delta) - V^\phi(\Delta+1) \\
= & \sum_{t=1}^{t_{max}}\left(\frac{p_t-p_t(1-p)^{t}}{p} + p_tp(1-p)^{t-1}\sigma\right).
\end{align*}
After some algebraic manipulations, we obtain
\[
\sigma = \ddfrac{\sum_{t=1}^{t_{max}}p_t\left(\frac{1-(1-p)^{t}}{p}\right)}{1-\sum_{t=1}^{t_{max}}pp_t(1-p)^{t-1}}.
\]
Then, we show that $V^{\phi}(\Delta+1) - V^{\phi}(\Delta)$ is independent of $\Delta\geq1$ under \textbf{Assumption 2}. Following the same steps, we can prove the desired results by mathematical induction. We first notice that, for each $\Delta\geq1$, the estimated value function is updated following \eqref{eq-EquivalentEq11}.
\begin{figure*}[!t]
\normalsize
\begin{multline}\label{eq-EquivalentEq11}
V^{\phi}_{\nu+1}(\Delta) = C(\Delta,1) - ET\theta^{\phi} +\sum_{t=1}^{t_{max}}\left[p_t\left( \sum_{k=0}^{t-1}P^t_{\Delta,k}(1)V_{\nu}^{\phi}(k) + P^t_{\Delta,\Delta+t}(1)V_{\nu}^{\phi}(\Delta+t)\right)\right] + \\
p_{t^+}\left( \sum_{k=0}^{t_{max}-1}P^{t^+}_{\Delta,k}(1)V_{\nu}^{\phi}(k) + P^{t^+}_{\Delta,\Delta+t_{max}}(1)V_{\nu}^{\phi}(\Delta+t_{max})\right).
\end{multline}
\hrulefill
\vspace*{4pt}
\end{figure*}
Meanwhile, the base case $\nu=0$ is true by initialization. Then, we assume $V^\phi_\nu(\Delta+1) - V^\phi_\nu(\Delta)=\sigma_\nu$ where $\sigma_\nu$ is independent of $\Delta\geq1$, and exam whether $V^\phi_{\nu+1}(\Delta+1) - V^\phi_{\nu+1}(\Delta)$ is independent of $\Delta\geq1$. Leveraging the properties in Lemma~\ref{lem-MSTPAss2}, we have
\begin{multline*}
V^{\phi}_{\nu+1}(\Delta+1) - V^{\phi}_{\nu+1}(\Delta)= C(\Delta+1,1) - C(\Delta,1) + \\
\left(\sum_{t=1}^{t_{max}}p_tP^t_{\Delta,\Delta+t}(1) + p_{t^+}P^{t^+}_{\Delta,\Delta+t_{max}}(1)\right)\sigma_{\nu}^{\phi}.
\end{multline*}
Moreover, according to the expressions in Lemma~\ref{lem-MSTPAss2}, we obtain
\begin{multline*}
\sum_{t=1}^{t_{max}}p_tP^t_{\Delta,\Delta+t}(1)+ p_{t^+}P^{t^+}_{\Delta,\Delta+t_{max}}(1) =\\
\sum_{t=1}^{t_{max}}p_tp(1-p)^{t-1} + p_{t^+}(1-p)^{t_{max}},
\end{multline*}
which is independent of $\Delta\geq1$. Leveraging the expression of $C(\Delta,1)$ in Lemma~\ref{lem-CompactCost}, we obtain
\begin{multline*}
C(\Delta,1) - C(\Delta-1,1) =\\
\sum_{t=1}^{t_{max}}p_t\bigg(\frac{1-(1-p)^t}{p}\bigg) + p_{t^+}\bigg(\frac{1-(1-p)^{t_{max}}}{p}\bigg).
\end{multline*}
We notice that $C(\Delta,1) - C(\Delta-1,1)$ is also independent of $\Delta\geq1$. Consequently, we can conclude that $V^{\phi}_{\nu+1}(\Delta+1) - V^{\phi}_{\nu+1}(\Delta)$ is independent of $\Delta\geq1$. Then, by mathematical induction, $V^{\phi}(\Delta+1) - V^{\phi}(\Delta)$ is independent of $\Delta\geq1$. We denote the constant by $\sigma$, which satisfies the following equation.
\begin{multline*}
\sigma = \sum_{t=1}^{t_{max}}p_t\bigg(\frac{1-(1-p)^t}{p}\bigg) + p_{t^+}\bigg(\frac{1-(1-p)^{t_{max}}}{p}\bigg) +\\
\left(\sum_{t=1}^{t_{max}}p_tp(1-p)^{t-1} + p_{t^+}(1-p)^{t_{max}}\right)\sigma.
\end{multline*}
After some algebraic manipulations, we obtain
\[
\sigma = \ddfrac{\sum_{t=1}^{t_{max}}p_t\left(\frac{1-(1-p)^t}{p}\right) + p_{t^+}\left(\frac{1-(1-p)^{t_{max}}}{p}\right)}{1-\left(\sum_{t=1}^{t_{max}}p_tp(1-p)^{t-1} + p_{t^+}(1-p)^{t_{max}}\right)}.
\]
\end{proof}
With Lemma~\ref{lem-VDeltaproperty} in mind, we can continue to the next step.

\paragraph{Policy Improvement} Here, we show that the new policy induced from the $V^\phi(\Delta)$ obtained in the previous step and $\theta^\phi$ is the threshold policy with $\tau=1$. To this end, we define $\delta V^\phi(\Delta)\triangleq V^{\phi,0}(\Delta)-V^{\phi,1}(\Delta)$, where $V^{\phi,a}(\Delta)$ is the value function resulting from taking action $a$ at state $(\Delta,0,-1)$. Then, the suggested action at state $(\Delta,0,-1)$ is $a=1$ if $\delta V^\phi(\Delta)\geq0$. Otherwise, $a=0$ is suggested. In the following, we investigate the expression of $\delta V^\phi(\Delta)$. We first notice that, for $\Delta\geq1$, $V^{\phi}(\Delta) = V^{\phi,1}(\Delta)$. Then, using Lemma~\ref{lem-VDeltaproperty}, we obtain
\begin{align*}
\delta V^\phi(\Delta) = & \Delta - \theta^\phi + (1-p)V^\phi(\Delta+1) + pV(0) - V^{\phi,1}(\Delta)\\
= & \Delta - \theta^\phi + (1-p)V^\phi(\Delta+1) + pV(0) - V^\phi(\Delta)\\
= & \Delta - 2\theta^\phi + [(1-p)-p(\Delta-1)]\sigma,
\end{align*}
where $\Delta\geq1$. We notice that
\[
\delta V^\phi(\Delta+1) - \delta V^\phi(\Delta) = 1 - p\sigma.
\]
For \textbf{Assumption 1}, plugging in the expression of $\sigma$ yields
\begin{align*}
1 - p\sigma = & 1- \ddfrac{\sum_{t=1}^{t_{max}}(p_t-p_t(1-p)^{t})}{1-\sum_{t=1}^{t_{max}}p_tp(1-p)^{t-1}}\\
= & \ddfrac{1-\sum_{t=1}^{t_{max}}p_tp(1-p)^{t-1} - \sum_{t=1}^{t_{max}}(p_t-p_t(1-p)^{t})}{1-\sum_{t=1}^{t_{max}}p_tp(1-p)^{t-1}}\\
= & \ddfrac{(1-2p)\sum_{t=1}^{t_{max}}p_t(1-p)^{t-1}}{1-\sum_{t=1}^{t_{max}}p_tp(1-p)^{t-1}}\geq0.
\end{align*}
For \textbf{Assumption 2}, we have
\begin{align*}
1 - p\sigma = & 1- \ddfrac{\sum_{t=1}^{t_{max}}p_t(1-(1-p)^t) + p_{t^+}(1-(1-p)^{t_{max}})}{1-\left(\sum_{t=1}^{t_{max}}p_tp(1-p)^{t-1} + p_{t^+}(1-p)^{t_{max}}\right)}\\
\geq & 1- \ddfrac{\sum_{t=1}^{t_{max}}p_t(1-(1-p)^t) + p_{t^+}(1-(1-p)^{t_{max}})}{1-\left(\sum_{t=1}^{t_{max}}p_t(1-p)^{t} + p_{t^+}(1-p)^{t_{max}}\right)}\\
=& 0.
\end{align*}
Consequently, when $\Delta\geq1$, $\delta V^\phi(\Delta+1) \geq \delta V^\phi(\Delta)$ for both assumptions. We notice that $\delta V^\phi(1) = 1 - 2\theta^\phi + (1-p)\sigma$. According to Condition~\ref{con}, $\theta^{\phi} = \bar{\Delta}_1\leq\frac{1+(1-p)\sigma}{2}$. Hence, we have
\[
\delta V^\phi(1) = 1 - 2\bar{\Delta}_1 + (1-p)\sigma\geq0.
\]
Combining together, we have
\[
\delta V^{\phi}(\Delta) \geq \delta V^{\phi}(1) \geq0\quad\Delta\geq1.
\]
Hence, the suggested action at state $(\Delta,0,-1)$ where $\Delta\geq1$ is to initiate the transmission (i.e., $a=1$). Now, the only missing part is the action at state $(0,0,-1)$. To determine the action, we recall from Theorem~\ref{thm-PIT} that the new policy will always be no worse than the old one. Meanwhile, by Condition~\ref{con}, $\bar{\Delta}_1\leq\bar{\Delta}_0$. Hence, the suggested action at state $(0,0,-1)$ is to stay idle (i.e., $a=0$). Combining with the suggested actions at other states, we can conclude that the policy improvement step yields the threshold policy with $\tau=1$.

Consequently, the policy iteration algorithm converges. Then, according to Theorem~\ref{thm-PIT}, we can conclude that the threshold policy with $\tau=1$ is optimal.

\end{document}